\documentclass[12pt]{article}

\usepackage{amsmath}
\usepackage{amsfonts}
\usepackage{amsthm}
\usepackage[usenames]{color}
\usepackage{setspace}
\usepackage{mathrsfs}
\usepackage{amssymb}
\usepackage{sgame}
\usepackage{tikz}
\usetikzlibrary{arrows}
\usetikzlibrary{arrows.meta}
\usepackage{multirow}
\usepackage[round]{natbib}
\usepackage{graphics}
\theoremstyle{definition}
\newtheorem{example}{Example}
\theoremstyle{definition}

\theoremstyle{definition}

\theoremstyle{plain}

\theoremstyle{definition}
\newtheorem{definition}{Definition}
\theoremstyle{definition}

\theoremstyle{plain}
\newtheorem{theorem}{Theorem}
\theoremstyle{plain}
\newtheorem{proposition}{Proposition}
\theoremstyle{plain}
\newtheorem{lemma}{Lemma}
\theoremstyle{plain}

\theoremstyle{definition}

\theoremstyle{remark}
\newtheorem{remark}{Remark}
\theoremstyle{definition}

\theoremstyle{definition}

\theoremstyle{definition}

\theoremstyle{plain}

\newtheorem{innercustomthm}{Proposition}
\newenvironment{customthm}[1]
  {\renewcommand\theinnercustomthm{#1}\innercustomthm}
  {\endinnercustomthm}

\newcommand{\red}[1]{{\color{red}#1}}
\newcommand{\blue}[1]{{\color{blue}#1}}

\newcommand{\abs}[1]{|#1|}
\newcommand{\set}[1]{\left\{#1\right\}}

\newcommand{\ceiling}[1]{\lceil#1\rceil}
\newcommand{\Ceiling}[1]{\Big\lceil#1\Big\rceil}
\newcommand{\flooring}[1]{\lfloor#1\rfloor}

\DeclareMathOperator*{\argmax}{argmax}

\newcommand{\shareNumber}{\kappa}

\newcommand{\neighbourhood}[1]{\ensuremath{N_{#1}}}
\newcommand{\degree}[1]{\ensuremath{d_{#1}}}

\newcommand{\strategySet}[1]{\ensuremath{X_{#1}}}
\newcommand{\strategy}[1]{\ensuremath{x_{#1}}}

\newcommand{\nominateSet}[2]{\ensuremath{M_{#1}^{#2}}}
\newcommand{\nominate}[1]{\ensuremath{m_{#1}}}

\newcommand{\Bcal}{\mathcal{B}}
\newcommand{\bestResponse}[2]{\Bcal_{#1}({#2})}
\newcommand{\bestReply}[2]{\Bcal_{#1}({#2})}

\newcommand{\Natural}{\mathbb{N}}
\newcommand{\Real}{\mathbb{R}}

\newcommand{\xbar}{\bar{x}}

\newcommand{\sbold}{\boldsymbol{x}}
\newcommand{\sboldstar}{\sbold^{*}}

\newcommand{\Sbold}{\boldsymbol{X}}

\newcommand{\nominateBold}{\boldsymbol{m}}
\newcommand{\nominateBoldStar}{\nominateBold^{*}}

\newcommand{\utility}[1]{\ensuremath{U_{#1}}}

\newcommand{\qstar}{q^{*}}

\newcommand{\biClique}{bi-clique neighbourhood graph}

\usepackage[pdftex, citecolor=blue, colorlinks]{hyperref}
\usepackage[margin=3cm]{geometry}

\begin{document}

\title{ \huge Public goods in networks with constraints on sharing}
\author{Stefanie Gerke\footnote{Mathematics Department, Royal Holloway University of London, Egham TW20 0EX, UK.} \and Gregory Gutin\footnote{Computer Science Department, Royal Holloway University of London, Egham TW20 0EX, UK.} \and Sung-Ha Hwang\footnote{College of Business, Korea Advanced Institute of Science and Technology (KAIST), Seoul, Korea.} \and Philip R.\ Neary\footnote{Economics Department, Royal Holloway University of London, Egham TW20 0EX, UK.}}

%\protect\footnote{We would like to thank Nizar Allouch, Jesper Bagger, Claudia Cerrone, Francesco Feri, Eoin Hanrahan, Roisin Hanrahan, Miguel \'{A}ngel Mel\'{e}ndez Jim\'{e}nez, David R.\ Neary, Frances Ruane, Nora Sadler, Joel Sobel, and Anders Yeo for helpful discussions. We are extremely grateful to the Associate Editor and two anonymous referees who made many excellent suggestions that improved the paper immeasurably. Any errors are ours.}

%\protect\footnote{Thanks to people / seminar audiences / grants / etc.}

\date{\today}

\maketitle

%%%% ----------------------------------------------------------------------
%\bigskip
%%%% ----------------------------------------------------------------------

\linespread{1.3}
%\doublespacing

%%% ----------------------------------------------------------------------
%%% ----------------------------------------------------------------------
%%% ----------------------------------------------------------------------
%%% ----------------------------------------------------------------------
%%% ----------------------------------------------------------------------
%%% ----------------------------------------------------------------------

\begin{abstract}
\noindent
%Building on the work of \cite{BramoulleKranton:2007:JET}, t
This paper considers incentives to provide goods that are partially shareable along social links.
%This paper considers incentives to provide goods that are partially excludable along social links.
%We introduce a model in which (i) a network game individuals decide how much of the good to provide, and (ii) endogenous network formation (individuals decide which subset of neighbours to nominate as co-beneficiaries).
We introduce a model in which each individual in a social network not only decides how much of a shareable good to provide, but also decides which subset of neighbours to nominate as co-beneficiaries.
An outcome of the model specifies an endogenously generated subnetwork of the original network and a public goods game occurring over the realised subnetwork.
We prove the existence of \emph{specialised} pure strategy Nash equilibria:\ those in which some individuals contribute while the remaining individuals free ride.
We then consider how the set of efficient specialised equilibria vary as the constraints on sharing are relaxed and we show that, paradoxically, an increase in shareability may decrease efficiency.
%We then consider how the set of efficient specialised equilibria vary as the constraints on sharing are relaxed and we show a monotonicity result.
%Finally, we introduce dynamics and show that only specialised equilibria can be stable against individuals unilaterally changing their provision level.
\end{abstract}

%(those with minimal number of Drivers)

%The proof is constructive and corresponds to showing, for a given capacity, the existence of a new kind of \emph{spanning bipartite subgraph}, a $DP$-\emph{subgraph}, with partite sets $D$ and $P$.

%For many [capacity, graph] pairs, there is a unique pure strategy Nash equilibrium which is unusual for a model of local public good provision (it also implies that certain individuals will always benefit maximally).

%Our model has two typically incompatible ingredients:\ (i) a graphical game (individuals decide how much of the good to provide), and (ii) graph formation (individuals decide which subset of neighbours to nominate as co-beneficiaries).

%%% ----------------------------------------------------------------------
%%% ----------------------------------------------------------------------
%%% ----------------------------------------------------------------------
%%% ----------------------------------------------------------------------
%%% ----------------------------------------------------------------------
%%% ----------------------------------------------------------------------

\newpage

%%% ----------------------------------------------------------------------
%%% ----------------------------------------------------------------------
%%% ----------------------------------------------------------------------
%%% ----------------------------------------------------------------------
%%% ----------------------------------------------------------------------
%%% ----------------------------------------------------------------------

\setstretch{1.3}

\section{Introduction}\label{INTRODUCTION}

Economists have long been aware of challenges that arise with the provision of public goods.
Despite the benefits of increased provision, the fact that the good is shareable can generate tension since all parties arrives at the same conclusion:\ ``I would rather someone else pays the cost of providing''.
\cite{BramoulleKranton:2007:JET}, hereafter BK, initiated the study of how societal structure might impact public good provision.\footnote{Other papers examine similar issues. See, for example, \cite{Allouch:2015:JET, Allouch:2017:GEB}, \cite{Baetz:2015:TE}, \cite{BramoulleKranton:2014:AER}, \cite{ElliottGolub:2019:JPE}, and \cite{KinatederMerlino:2017:AEJM}.}
Specifically, BK asked:\ does the social network governing who interacts with who have bearing on who will provide and who will not?
By showing that societal structure not only facilitates but also encourages {\it specialisation} (i.e., outcomes in which some individuals contribute while the remaining individuals free ride), BK confirmed that the answer to this question is a resounding ``yes''.

%BK demonstrated that networks encourage {\it specialisation} in public good provision - outcomes in which some individuals contribute while the remaining individuals free ride - thereby confirming that the answer to the above question is a resounding ``yes''.

%BK confirmed that the answer to this question is a resounding yes by showing that networks facilitate specialisation:\ equilibrium outcomes in which some individuals contribute while the remaining individuals free ride.
%In this paper we extend the model of BK to allow for the possibility that one's contribution of the shareable good can be shared, it may not be possible to share with all of one's neighbours.

%In this paper, we extend the model of BK to allow for the possibility that it may not be possible for an individual to share their contribution with only some of their neighbours.
In this paper, we extend the BK model to allow for scenarios in which providers of the public good are limited in their ability to disseminate it.
That is, an individual may only share their contribution with a fixed number of neighbours, and precisely which set of neighbours must be specified.
Examples where constraints on sharing seem a more reasonable modelling choice include who to offer a ride in your five-seater car to (when the group of friends is six or more), and deciding which three family members to share your Netflix password with.
%Restrictions on sharing, and the many permutations of who decides to share with whom, can greatly impact the provision of shareable resources.
We focus on two main issues.
When there are restrictions placed sharing, is specialisation in public good provision still be consistent with equilibrium?
And if yes, do the set of specialised equilibrium outcomes change as shareability varies?
In particular, will an increase in the ability to share always lead to greater efficiency?

The simplest version of our model is termed the ``Netflix Game'' as it is inspired by the online streaming provider Netflix.
The quantity choice is binary:\ each individual in a social network $G$ decides either to purchase a Netflix account or not (0 or 1).
If individual $i$ purchases an account, she then {\it nominates} exactly $\shareNumber(i)$ neighbours as co-beneficiaries, where $\shareNumber(i)$ is an exogenously given number known as $i$'s {\it sharing capacity} (henceforth capacity).\footnote{Two quick comments. First,  we assume throughout this section that $\shareNumber(i)$ is less than or equal to $i$'s number of neighbours. Second, in applications, for example with Netflix sharing, $\shareNumber(i)$ is often the same for all $i$.}
%\footnote{If $\shareNumber(i)$ is greater than $i$'s number of neighbours then $i$ must nominate all her neighbours. Thus, in what follows we will assume, unless specified otherwise, that the number of neighbours of each individual $i$ is not smaller than $\shareNumber(i)$. (Otherwise, one can simply reduce $\shareNumber(i)$ to be equal to $i$'s number of neighbours without changing the environment.)}
% and $\degree{G}(i)$ is $i$'s degree in $G$.\footnote{The reader unfamiliar with graph-theoretic terminology can skip ahead to the beginning of Section \ref{MODEL} for the formal definitions.}
%\footnote{If $\shareNumber(i)$ is greater than $i$'s number of neighbours then $i$ must nominate all her neighbours.}
% and $\degree{G}(i)$ is $i$'s degree in $G$.\footnote{The reader unfamiliar with graph-theoretic terminology can skip ahead to the beginning of Section \ref{MODEL} for the formal definitions.}
Preferences are such that it is better to have access to Netflix than not, even if that means paying for an account yourself.
However, due to the cost, it is preferable that a neighbour purchases an account and nominates you as a co-beneficiary than vice-versa.\footnote{\label{fn:nominate}A nomination is an offer of access that may or may not be exercised. In the binary action Netflix Game, each additional nomination yields no extra benefit. This seems reasonable to us:\ what can one do with access to two Netflix accounts that one cannot do with access to only one? In the general model that we introduce later, we allow for the possibility that strictly more of the good is always strictly beneficial and so any additional nomination is always exercised.}

%Given this, in any pure strategy Nash equilibrium all individuals have access to Netflix.

% (Similarly when offered two free rides to a destination, the second offer yields no additional benefit.)

%\footnote{There is then the issue of how to interpret the word \emph{nominate}. When we write ``individual $i$ nominates (neighbouring) individual $j$'', we mean that $i$ offers his action choice to $j$. In the general model, any offer will always be accepted. In the binary action Netflix game, each additional offer yields no extra benefit. As such, just as in the best-shot game, if an individual receives two or more offers only one of the offers will be `exercised' and it is not specified by the model which offer is exercised.}

%A strategy profile in the Netflix Game involves each player $i$ choosing to purchase or not and also nominating a fixed subset of neighbours of size $\shareNumber(i)$.
%A pure strategy profile in the Netflix Game involves each player choosing whether or not to purchase and nominating a fixed subset of neighbours.
A pure strategy in the Netflix Game specifies an endogenously generated subnetwork of the original network and a public goods game occurring over the realised subnetwork.
%\marginpar{\tiny{\PN{PN: Maybe move D and P to later? Careful that proof uses DP-Nash...}}}
Our focus is on pure strategy Nash equilibria, wherein every individual is one of two kinds:\ those who purchase a Netflix account, the $D$-set, or those who free ride, the $P$-set.
%\footnote{While our model was inspired by Netflix, there are many other situations to which the model is applicable. One such situation is that of ride sharing wherein everyone is either a Driver or a Passenger.}
%While a $DP$-Nash subgraph is purely graph-theoretic, it has an intuitive economic interpretation:\
At any pure strategy Nash equilibrium, each individual $i$ in $D$ nominates exactly $\shareNumber(i)$ neighbours while each individual in $P$ must be nominated (by at least one neighbour from $D$).\footnote{An offer of access to Netflix flows along the edges of the endogenously generated subnetwork from those in $D$ to those in $P$.}
% and no individual can be both in $D$ and in $P$
Our first result, Theorem \ref{theorem:balancedPSNE}, confirms the existence of a pure strategy Nash equilibrium for any Netflix Game.
%\marginpar{\tiny{\PN{PN: I removed the formal def of DP-Nash here and instead reference it.}}}
The proof is constructive and amounts to proving the existence of a novel type of spanning bipartite subgraph termed a $DP${\it-Nash subgraph} (see Definition~\ref{def:DPNash}).

%, for a given capacity function $\shareNumber$, the existence of a $DP$-Nash subgraph of $G$:

%\footnote{\label{fn:defDP}The proof is constructive and amounts to showing, for a given capacity function $\shareNumber$, the existence of a $DP$-Nash subgraph of $G$:\ a spanning bipartite subgraph $H$ of $G$ with partite sets $P$ and $D$ where for each $i \in D$ the degree of $i$ in $H$ is $\min\set{\shareNumber(i),\degree{G}(i)}$ and for every $i\in P$ the degree of $i$ in $H$ is positive.}

%Given its constructive nature, the proof suggests an algorithm that finds a pure strategy equilibrium in polynomial time.

%\footnote{Our model does not admit a potential function \citep{ShapleyMonderer:1996:GEB} which would render the existence of a pure strategy equilibrium immediate.}

%While the model is too rich for formal theories of equilibrium selection, w

In any pure strategy Nash equilibrium every individual must have access to Netflix (since if someone is without access, it is a best-response to buy access).
Given this, and given the fact that a Netflix account is not only costly but also costs the same for everyone, we can relate equilibria to efficiency.
%Given this and the fact
%A Netflix account is not only costly but costs the same for everyone.
%Given this and given the fact that in any pure strategy equilibrium everyone will have access (since anyone who does not have access cann, we can relate equilibria to efficiency.
We call an equilibrium {\it efficient} ({\it inefficient}) when the size of the $D$-set is smaller (greater) than any other equilibrium.
%At the other extreme, are the most inefficient equilibria - those with the largest $D$-set.
One might conjecture that these extreme efficient (inefficient) equilibria can only improve as the ability to share increases.
%
%sizes of both maximal and minimal $D$-sets are monotonically decreasing in capacity since at least as much sharing is possible.
We show via some examples that this is not the case.
% as certain efficient equilibria are fragile to minor changes in individual capacities. 
In particular, the most efficient equilibrium need not occur when capacity is maximal, i.e., when each individual has the ability to share with all of their neighbours.
This means that, paradoxically, restrictions on the ability to share can improve societal outcomes in a world where sharing bestows benefit on others.\footnote{This ``improvement'' is from the perspective of demand. On the (un-modelled) supply side, Netflix Inc.\ may prefer selling more accounts, in which case restricting sharing would reduce profits.}
However, for any two ordered capacity functions $\shareNumber$ and $\shareNumber'$ (i.e., $\shareNumber(i) \leq \shareNumber'(i)$ for all $i$), Theorem \ref{thm:delta} shows that the most efficient equilibrium for $\shareNumber'$ is never less efficient than the most inefficient equilibrium for $\shareNumber$.

With the above ideas fixed we can now introduce our full model.
The extension occurs along two dimensions:\ (i) quantity choice is no longer simply $0$ or $1$ but rather any non-negative integer, and (ii) preferences are now defined by a positive quantity, $\qstar$, at which an individual would always (at least weakly) benefit from more of the good but would never pay for more themselves, since the marginal benefit of each additional unit is nonincreasing whereas marginal cost is constant. ({In the binary action Netflix Game, $\qstar = 1$, and this amount fully satiates each individual. The interpretation being that no additional benefit is accrued from access to more than one Netflix account. Recall Footnote \ref{fn:nominate}.)

In the full model there may be pure strategy Nash equilibria in which each of two or more individuals contribute a positive quantity less than $\qstar$.
{A simple example would be a two-person network where each individual provides a positive quantity that they share with the other and these two quantities sum to $\qstar$.}
We follow BK and focus on so-called \emph{specialised} pure strategy Nash equilibria, wherein each contributing individual provides exactly $\qstar$ units of the good while those free-riding provide nothing.
%Specialised equilibria in the full model parallel pure strategy equilibria of the binary-action Netflix game.
A specialised equilibrium in the BK model coincides with the well-studied graph-theoretic notion of an {\it independent dominating set}, also called a maximal independent set (see \cite{GoddardHenning}).
That is, when sharing is required to be done up to capacity, no two neighbours can both supply $\qstar$ in equilibrium. 
This is intuitive since if one is already receiving $\qstar$ from a neighbour, one will not pay for any themselves because the cost outweighs the benefit.

Two neighbouring individuals can both supply $\qstar$ in the specialised equilibria of our model because individuals may be constrained in their ability to share.
The reason being that they need not share with each other (for an example see the specialised equilibrium depicted in Figure~\ref{fig:twoStarConnectedEquilibria} in which $I$ and $J$ are connected in the social network and yet both provide $\qstar$).
Mathematically, those who supply in the specialised equilibria of our model are described by $D$-sets, a generalisation of independent dominating sets, that exhibit two additional properties of interest.
%\marginpar{\tiny{\PN{I realised I had made a mistake here.}}}
First, no two independent dominating sets can exhibit set inclusion but it is possible for two $D$-sets to be so ordered.
When this occurs, the corresponding pair of equilibria can be ranked by the Pareto criterion.
%This allows us to relate equilibria to efficiency.
%\footnote{\PN{For the binary action Netflix Game this generates the curious ramification that there may be instances of the model with equilibria that are not Pareto efficient. Although this statement does not carry over to the general model where ``more is always better''.}}
Second, for some capacity functions a graph may possess a unique $D$-set.
This means that, in contrast to BK, there are instances of our model wherein there is a unique specialised equilibrium.\footnote{\cite{GutinNeary:2023:DAM} classify when our model admits a unique pure strategy equilibrium and show that the question of uniqueness raises interesting issues of computational complexity.}
That is, constraints on sharing can curb equilibrium multiplicity.

Following BK, we repeat the game and introduce best-response dynamics.
%{In particular, while non-specialised equilibria can be interpreted as equilibria wherein there is ``some cooperation / coordination \dots amongst the individuals of a given society'', we will see that such equilibria are, in a sense to be made precise in Section~\ref{DYNAMICS}, not stable.}
Because an individual's choice of nomination is only payoff relevant to others, the number of best-responses can be enormous.
As such, we start with the nomination component of equilibrium strategy profiles held fixed and we consider unilateral deviations only in action choice.\footnote{We emphasise that there is no clear game-theoretic defence of this assumption. Rather it seems to us (i) the simplest way to begin an analysis of dynamics (and therefore a sensible place to start), and (ii) consistent with the applied intuition that links are slower to adapt than behaviour.}
That is, from a particular equilibrium we force a change in one individual's action choice, and from there we ``let the system go'' tracking how population behaviour evolves as all individuals update the action choice every period.
An equilibrium is said to be stable if the dynamics return to it.
%That is, from a particular equilibrium we force a change in one individual's action choice and then we ``let the system go''; from that point on we track how population behaviour evolves.
Proposition \ref{proposition:necessary} shows that specialised equilibria are necessary for stability.
Proposition \ref{proposition:notSufficient} shows that specialised equilibria are not sufficient for stability, but, Theorem~\ref{theorem:sufficient} shows that a specialised equilibrium supported a by a $DP$-Nash subgraph satisfying a mild density condition is sufficient for stability.
%Theorem~\ref{theorem:sufficient} is the analog of Theorem 2 in BK.

%As such we restrict attention to \emph{nicely balanced specialised equilibria} - those in which every individual in $P$ nominates at least one of the individuals in $D$ who nominated them.

%We show that only nicely balanced specialised equilibria are locally stable to unilateral deviations in action choice. In particular non-specialised equilibria are not robust. Our Proposition \ref{proposition:sufficient} addressing this issue is the analog of Theorem 2 in BK.

%We show in Proposition \ref{proposition:necessary} that balanced specialised profiles are necessary for stability, the following result, Proposition \ref{proposition:notSufficient} shows that they are not sufficient. However, Proposition \ref{proposition:sufficient} shows that balanced specialised profiles supported by a $DP$-Nash subgraphs with a density condition are sufficient for stability.

%There are a host of existing papers that examine issues associated with local public good provision

%\footnote{}

Our model generalises that of BK to include a component of network formation.\footnote{Recently \cite{AllouchKing:2019:JPET} generalised the BK model in a different direction by allowing for constrained provision. Formally, this is done by extending payoffs to allow for the possibility that a person in isolation would contribute maximally and yet still have marginal benefit exceed marginal cost. Just as the model is a generalisation, so too is the equilibrium concept. The analog of specialised equilibria are given by \emph{insulated sets} \citep{JagotaNarasimhan:2001:DAM} that generalise independent dominating sets (albeit in a different way to $D$-sets).}
Despite this, in some ways our work is arguably closer to that of \cite{GaleottiGoyal:2010:AER}.
%In this sense our paper is close to that of \cite{GaleottiGoyal:2010:AER}.
Like us, they propose a model of public good provision in which individuals may nominate others in society.
The differences are that in their model (i) nominations are costly, and (ii) an individual benefits from the contributions of those who nominate him and also of those who he himself nominates.
This ability to \emph{piggy back} on the (potentially large) contributions of others propels the system towards orderly-looking, so called {\it core-periphery}, networks.
%\footnote{Moreover, the fraction of individuals contributing a positive quantity of the good tends to zero as population size grows.}
%This should be contrasted with our model wherein the structure of the endogenously generated subnetworks is not always the same; rather the social networks that arise in equilibrium are determined by the original underlying network and the specifics of the capacity function.
This should be contrasted with our model wherein the social networks that arise in equilibrium are determined by the original underlying network and the specifics of the capacity function. And these can vary substantially.

%structure of the endogenously generated subnetworks is not always the same; rather the social networks that arise in equilibrium are determined by the original underlying network and the specifics of the capacity function.

%\footnote{\cite{GaleottiGoyal:2010:AER} show that the fraction of agents contributing positive quantities of the good goes to zero as $n$ gets large.}

%helps the system rule out inefficient equilibria.\footnote{Here efficiency is the same as our usage (utilitarian). Formally, \cite{GaleottiGoyal:2010:AER} show that the fraction of agents contributing positive quantities of the good goes to zero as $n$ gets large.}
%Moreover, in the model of \cite{GaleottiGoyal:2010:AER} only orderly-looking, so called {\it core-periphery}, networks emerge in equilibrium. This should be contrasted with our model wherein the endogenously emerging network can take any form - with the form dictated by the details of the capacity function.

Before beginning the paper proper, we touch on an important modelling choice.
We require that each individual $i$ with capacity $\shareNumber(i)$ must nominate precisely $\shareNumber(i)$ neighbours and not some subset of neighbours whose size is no greater than $\shareNumber(i)$.\footnote{In such a modification, specialised equilibria correspond to {\it capacitated dominating sets} - a construct that has received attention in graph theory \citep{Guha:2003vc,CyganPW11,KaoCL15}. The specialised equilibria of our model are then referred to as {\it exact capacitated dominating sets}.}
We have two reasons for this, one mathematical and the other economic. 
The mathematical justification is straightforward. 
By insisting that capacity constraints are saturated our model is a generalisation of existing public goods problems on networks.
In particular, the ``limiting'' case of our model, that in which each individual's capacity equals their number of neighbours, \emph{is} the model of BK (since a $D$-set for a graph where every vertex has capacity equal to its number of neighbours \emph{is} an independent dominating set). 
Were we instead to allow individual $i$ to nominate any number of neighbours weakly less than $\shareNumber(i)$, this would not hold true and so it would be difficult to compare specialised equilibria of our model with those of BK.

The economic justification for requiring full saturation is more subjective.
It seems antithetical to public goods issues that individuals can be excluded (at least up until a capacity constraint binds).
Given that sharing is costless it seems natural to us that ``full nomination'' will occur.
And once we assume that individuals are willing to voluntarily share resources, why not assume they do so maximally?
In Section~\ref{NomDYNAMICS}, we briefly consider the variant of the model in which each individual $i$ may nominate up to but no more than $\shareNumber(i)$ neighbours. 
For a large class of networks, that includes many of those considered in BK,  we show that only specialised equilibria in which every individual saturates their sharing capacity are stable to deviations in nomination.

We conclude the introduction by relating our model to those of games on endogenous networks.\footnote{See, amongst others, \cite{Baetz:2015:TE}, \cite{Cho:2010:IJGT}, \cite{GoyalJabbari:2016:}, \cite{GoyalVega-Redondo:2005:GEB}, \cite{HojmanSzeidl:2006:GEB} \cite{KinatederMerlino:2017:AEJM}, \cite{SadlerGolub:2021:arXiv}, \cite{Staudigl:2011:GEB}, and \cite{StaudiglWeidenholzer:2014:JET}.}
%We highlight that a model of endogenous network formation followed by a public goods game over the realised network is in fact a special case of our model.
One interpretation of our model is that the exogenously specified network represents constraints on who each individual may nominate.
Viewed in this light, when the exogenous network is a complete graph, each individual $i$ is unconstrained and so may nominate any subset of individuals of size $\shareNumber(i)$.
Thus, a public goods game in endogenous networks is but a special case of our model.
We discuss this interpretation briefly in Appendix \ref{APP:endogenous}.
Amongst other things, we show that our equilibrium existence result (Theorem \ref{theorem:balancedPSNE}) goes through under this interpretation.

The remainder of the paper is organised as follows. Section \ref{EXAMPLE} motivates our analysis with three examples. Section \ref{MODEL} introduces the model and proves existence of a specialised Nash equilibrium for every instance of the model.
Section \ref{COMPARATIVESTATICS} examines comparative statics and efficiency.
Section \ref{DYNAMICS} introduces dynamics and shows that specialised equilibria are necessary for stability. Section \ref{CONCLUSION} concludes with a summary of our results and some suggestions for further research on this topic.

\section{Examples}\label{EXAMPLE}

%%% ----------------------------------------------------------------------
%%% ----------------------------------------------------------------------

This section discusses three examples that illustrate features of the model and highlight some of our main results.
The first example shows that the set of specialised equilibria in a world of local public good provision can change dramatically with the introduction of capacity constraints.
%In particular the most efficient equilibrium (defined as that with the smallest $D$-set) can disappear with only a minor change in capacities.
The second example shows how the size of $D$-sets may evolve non-monotonically as the constraints on capacity are relaxed.
In particular, and in our view paradoxically, the most efficient equilibrium outcomes may not occur when the ability to share is maximal.
The third example previews Section~\ref{DYNAMICS} on dynamic stability. 
This example illustrates how pure strategy equilibria that are not specialised can quickly unravel when even one individual changes their action choice.

%\footnote{A specialised equilibrium is one in which each player contributes either zero or the locally optimal positive quantity.}

\begin{example}[Netflix provision]\label{ex:Netflix}
There is a social network of $5$ individuals arranged in a star as depicted in Figure \ref{fig:5star}. We label the peripheral individuals by $h, i, j$, and $k$ and the central individual by $\ell$.

%\begin{figure}[ht!]
%\centering
%\begin{tikzpicture}
%
%\tikzset{vertex/.style = {shape=circle,draw,minimum size=1.5em}}
%\tikzset{edge/.style = {->,> = latex'}}
%% vertices
%\node[vertex] (a) at  (0,3) {$h$};
%\node[vertex] (b) at  (0,0) {$k$};
%\node[vertex] (c) at  (3,0) {$j$};
%\node[vertex] (d) at  (3,3) {$i$};
%
%\node[vertex] (e) at  (1.5, 1.5) {$\ell$};
%
%
%%edges
%\draw (a) edge (e);
%
%
%\draw (e) edge (b);
%\draw (c) edge (e);
%
%\draw (e) edge (d);
%
%
%%\draw[edge] (a)  to[bend left] (a1);
%%\draw[edge] (a1) to[bend left] (a);
%
%%\draw[edge] (a1) to[bend left] (a2);
%%\draw[edge] (a2) to[bend left] (a1);
%%
%%\path (a2) to node {\dots} (c);
%%\node [shape=circle,minimum size=1.5em] (a3) at (4.5,0) {};
%%\draw[edge] (a2) to[bend left] (a3);
%%\draw[edge] (a3) to[bend left] (a2);
%
%%\node [shape=circle,minimum size=1.5em] (c1) at (6.5,0) {};
%%\draw[edge] (c) to[bend left] (c1);
%%\draw[edge] (c1) to[bend left] (c);
%\end{tikzpicture}
%\caption{A 5-person star network}
%\label{fig:5star}
%\end{figure}

\begin{figure}[hbt!]
\centering
\tikzstyle{vertexX}=[circle,draw, top color=gray!10, bottom color=gray!70, minimum size=10pt, scale=0.9, inner sep=0.5pt]
\tikzstyle{vertexY}=[circle,draw, top color=black!10, bottom color=gray!70, minimum size=25pt, scale=0.8, inner sep=0.4pt]
\begin{tikzpicture}[scale=0.5]
%\draw (5,-1.5) node {{\small $(G_1,\kkk_1)$}};
%\draw (1.0,1.0) node {{\small $1$}};
%\draw (1.0,6.0) node {{\small $1$}};
%\draw (9.0,6.0) node {{\small $1$}};
%\draw (9.0,1.0) node {{\small $1$}};
%\draw (5.0,4.9) node {{\small $3$}};
%\node (x) at (3.0,1.0) [vertexY] {$x$}; 
%\node (z) at (7.0,1.0) [vertexY] {$z$}; 
\node (y) at (5.0,5.0) [vertexY] {$\ell$};
\node (x1) at (2.5,2.5) [vertexY] {$k$}; 
\node (x2) at (2.5,7.5) [vertexY] {$h$};
\node (x3) at (7.5,7.5) [vertexY] {$i$}; 
\node (x4) at (7.5,2.5) [vertexY] {$j$}; 
\draw [thick] (y) -- (x1); 
\draw [thick] (y) -- (x2); 
\draw [thick] (y) -- (x3); 
\draw [thick] (y) -- (x4); 
%\draw [thick] (z) -- (y2); 
%\draw [thick] (z) -- (y3); 
\end{tikzpicture}
\caption{A 5-person star network.}\label{fig:5star}
\end{figure}

Each individual can purchase a Netflix account or not.
Buying a Netflix account costs $c > 0$ but brings a benefit that exceeds this.
%benefitwishes to utilise the online media services provider Netflix.
Not having access to Netflix leaves each individual with a payoff of zero.
The company's current rules permit any individual who purchases an account to stream simultaneously on a maximum of five devices.
We assume that edges in the network represent close friendships so that any person who purchases can and will share with each of his friends.
Formally, this is modelled as the simultaneous-move, so-called, best shot game of \cite{GaleottiGoyal:2010:RES} where each agent has strategy set $\set{0, 1}$, with $1$ meaning purchase a Netflix account and $0$ meaning don't.\footnote{In the same way that our general model generalises BK's model, the binary action Netflix Game generalises the best-shot game of \cite{GaleottiGoyal:2010:RES}.}

The best case scenario for each player is that a neighbour purchases a Netflix account.
This is optimal because then one has access to Netflix without paying the cost.
Given this, it is a best-response for each player to purchase a Netflix account if and only if no neighbour does.
In this set up there are two pure strategy equilibria.
In the first, only the central individual, $\ell$, purchases. In the second, all the peripheral individuals, $h, i, j$, and $k$ purchase and $\ell$ does not.
These two equilibria are depicted in Figure \ref{fig:starEquilibria} below, with buyers in \blue{blue} and non-buyers in \red{red}.
The direction of sharing is indicated by arrows, with the tail of any arrow originating at a buyer and the head of an arrow pointing to those with whom the buyers offers to share access.
In both equilibria the set of adopters forms an independent dominating set so both equilibria are Pareto-efficient.

%\begin{figure}[ht!]
%\centering
%\begin{tikzpicture}
%
%\tikzset{vertex/.style = {shape=circle,draw,minimum size=1.5em}}
%\tikzset{edge/.style = {->,> = latex'}}
%% vertices
%\node[fill = red, vertex] (a) at  (0,3) {$h$};
%\node[fill = red, vertex] (b) at  (0,0) {$k$};
%\node[fill = red, vertex] (c) at  (3,0) {$j$};
%\node[fill = red, vertex] (d) at  (3,3) {$i$};
%%
%\node[fill = cyan, vertex] (e) at  (1.5, 1.5) {$\ell$};
%
%\draw[edge] (e) to (a);
%\draw[edge] (e) to (b);
%\draw[edge] (e) edge (c);
%\draw[edge] (e) to (d);
%
%%\drawedge(e,d)
%
%%------------------------
%
%\node[fill = cyan, vertex] (a1) at  (6,3) {$h$};
%\node[fill = cyan, vertex] (b1) at  (6,0) {$k$};
%\node[fill = cyan, vertex] (c1) at  (9,0) {$j$};
%\node[fill = cyan, vertex] (d1) at  (9,3) {$i$};
%%
%\node[fill = red, vertex] (e1) at  (7.5, 1.5) {$\ell$};
%
%%edges
%\draw[edge] (a1) to (e1);
%\draw[edge] (b1) to (e1);
%\draw[edge] (c1) to (e1);
%\draw[edge] (d1) to (e1);
%
%
%
%\end{tikzpicture}
%\caption{Equilibria for best-shot game on 5-person star network}
%\label{fig:starEquilibria}
%\end{figure}

\begin{figure}[hbt!]
\centering
\tikzstyle{vertexX}=[circle,draw, top color=gray!10, bottom color=gray!70, minimum size=10pt, scale=0.9, inner sep=0.5pt]
\tikzstyle{vertexY}=[circle,draw, top color=black!10, bottom color=red!70, minimum size=25pt, scale=0.8, inner sep=0.4pt]
\tikzstyle{vertexZ}=[circle,draw, top color=black!10, bottom color=blue!70, minimum size=25pt, scale=0.8, inner sep=0.4pt]
\tikzset{arc/.style = {->,> = latex'}}
\hfill
\begin{tikzpicture}[scale=0.5]
%\draw (5,-1.5) node {{\small $(G_1,\kkk_1)$}};
%\draw (1.0,1.0) node {{\small $1$}};
%\draw (1.0,6.0) node {{\small $1$}};
%\draw (9.0,6.0) node {{\small $1$}};
%\draw (9.0,1.0) node {{\small $1$}};
%\draw (5.0,4.9) node {{\small $3$}};
%\node (x) at (3.0,1.0) [vertexY] {$x$}; 
%\node (z) at (7.0,1.0) [vertexY] {$z$}; 
\node (y) at (5.0,5.0) [vertexZ] {$\ell$};
\node (x1) at (2.5,2.5) [vertexY] {$k$}; 
\node (x2) at (2.5,7.5) [vertexY] {$h$};
\node (x3) at (7.5,7.5) [vertexY] {$i$}; 
\node (x4) at (7.5,2.5) [vertexY] {$j$}; 
%\draw [thick] (y) -- (x1); 
%\draw [thick] (y) -- (x2); 
%\draw [thick] (y) -- (x3); 
%\draw [thick] (y) -- (x4); 
%\node (x1) at (2.5,1.0) [vertexY] {$x_{1}$}; 
%%\node (x2) at (2.5,6.0) [vertexY] {$x_{1}$};
%\node (x3) at (5.0,6.6) [vertexY] {$x_{2}$}; 
%\node (x4) at (7.5,1.0) [vertexY] {$x_{3}$}; 
%\node (y) at (5.0,3.5) [vertexZ] {$x_{4}$};
\draw [arc,line width=1.3pt] (y) -> (x1); 
\draw [arc,line width=1.3pt] (y) -> (x2); 
\draw [arc,line width=1.3pt] (y) -> (x3); 
\draw [arc,line width=1.3pt] (y) -> (x4); 
%\draw [thick] (y) -- (x1); 
%\draw [thick] (y) -- (x2); 
%\draw [thick] (y) -- (x3); 
%\draw [thick] (y) -- (x4); 
%\draw [thick] (z) -- (y2); 
%\draw [thick] (z) -- (y3); 
\end{tikzpicture} \hfill 
\begin{tikzpicture}[scale=0.5]
\node (y) at (5.0,5.0) [vertexY] {$\ell$};
\node (x1) at (2.5,2.5) [vertexZ] {$k$}; 
\node (x2) at (2.5,7.5) [vertexZ] {$h$};
\node (x3) at (7.5,7.5) [vertexZ] {$i$}; 
\node (x4) at (7.5,2.5) [vertexZ] {$j$}; 
%\draw [thick] (y) -- (x1); 
%\draw [thick] (y) -- (x2); 
%\draw [thick] (y) -- (x3); 
%\draw [thick] (y) -- (x4); 
%\node (x1) at (2.5,1.0) [vertexY] {$x_{1}$}; 
%%\node (x2) at (2.5,6.0) [vertexY] {$x_{1}$};
%\node (x3) at (5.0,6.6) [vertexY] {$x_{2}$}; 
%\node (x4) at (7.5,1.0) [vertexY] {$x_{3}$}; 
%\node (y) at (5.0,3.5) [vertexZ] {$x_{4}$};
\draw [arc,line width=1.3pt] (x1) -> (y); 
\draw [arc,line width=1.3pt] (x2) -> (y); 
\draw [arc,line width=1.3pt] (x3) -> (y); 
\draw [arc,line width=1.3pt] (x4) -> (y); 
%\draw (5,-1.5) node {{\small $(G_1,\kkk_1)$}};
%\draw (1.0,1.0) node {{\small $1$}};
%\draw (1.0,6.0) node {{\small $1$}};
%\draw (9.0,6.0) node {{\small $1$}};
%\draw (9.0,1.0) node {{\small $1$}};
%\draw (5.0,4.9) node {{\small $3$}};
%\node (x) at (3.0,1.0) [vertexY] {$x$}; 
%\node (z) at (7.0,1.0) [vertexY] {$z$}; 
%\node (x1) at (2.5,1.0) [vertexZ] {$x_{1}$}; 
%%\node (x2) at (2.5,6.0) [vertexY] {$x_{1}$};
%\node (x3) at (5.0,6.6) [vertexZ] {$x_{2}$}; 
%\node (x4) at (7.5,1.0) [vertexZ] {$x_{3}$}; 
%\node (y) at (5.0,3.5) [vertexY] {$x_{4}$};
%%\draw [arc,line width=1.3pt] (y) -> (x1); 
%
%\draw [thick] (y) -- (x1); 
%%\draw [thick] (y) -- (x2); 
%\draw [thick] (y) -- (x3); 
%\draw [thick] (y) -- (x4); 
%%\draw [thick] (z) -- (y2); 
%%\draw [thick] (z) -- (y3); 
\end{tikzpicture} \hfill \hfill 
\caption{Equilibria for 5-person star network with full sharing}
\label{fig:starEquilibria}
\end{figure}

%\begin{figure}[ht!]
%\centering
%\includegraphics{Images/starEquilibria.pdf}
%\caption{Equilibria for best-shot game on 5-person star network}
%\label{fig:starEquilibria}
%\end{figure}

Now let us consider what would happen if Netflix altered the number of devices that one account can simultaneously access. Let the number of people that may simultaneously use the service other than the account holder be denoted by $\shareNumber$. An individual who purchases an account can nominate only $\shareNumber$ of her neighbours (and will nominate all of her neighbours if she has less than $\shareNumber$). For $\shareNumber = 1, 2$, and $3$, the only equilibrium outcome is for the peripheral players to purchase an account (and each to nominate the central individual, $\ell$, as the friend who may use the account free of charge).
These equilibrium outcomes have a $D$-set of size 4 and are each depicted in the right hand panel of Figure \ref{fig:starEquilibria}.

The outcome depicted in the left hand panel is no longer supported by an equilibrium.
The reason is that if $\ell$ buys an account then she can only offer to share with 3 of her 4 neighbours which will leave one, say $i$, without access.
It is then optimal for the un-nominated $i$ to buy an account. But $i$ will then offer to share access with $\ell$ who subsequently would no longer want to buy. So in this example, restricting attention to equilibria, the number of adopters in equilibrium has decreased as Netflix Inc.\ allow more ``shareability''. %Finally, we note that instances of our model admit a unique equilibrium.% The following example will show that this is not always the case.
\end{example}

\begin{example}[Netflix provided efficiently]\label{ex:twoStars}
The previous example suggests a natural conjecture:\ that a reduction in the ability to share can only increase the number of Netflix accounts that are sold in equilibrium.
%This is, loosely put, that the size of $D$-sets will not increase as capacity increases. 
The following example refutes the conjecture and also highlights some other interesting features of the set up.

%A second conjecture might be that the number of equilibrium outcomes increases as the ability to share increases.

There is a social network of six individuals arranged as depicted in Figure \ref{fig:twoStarConnected}. We label the two central individuals by $I$ and $J$ and the four peripheral individuals by $i_{1}, i_{2}, j_{1}$, and $j_{2}$ (where the letter index on the peripheral individuals indicates which of the two central individuals they are linked to).
%
%We imagine a couple, Isabella ($I$) and John ($J$).
%Isabella's has three friends, $i_{1}, i_{2}$, and $i_{3}$, and John has three friends, $j_{1}, j_{2}$, and $j_{3}$.
%Isabella knows her three friends and John knows his three friends; no other pair of individuals have previously met. This social network is depicted in Figure \ref{fig:twoStarConnected}. 

\begin{figure}[ht!]
\centering
\begin{tikzpicture}

\tikzstyle{vertexX}=[circle,draw, top color=gray!10, bottom color=gray!70, minimum size=10pt, scale=0.9, inner sep=0.5pt]
\tikzstyle{vertexY}=[circle,draw, top color=black!10, bottom color=gray!70, minimum size=25pt, scale=0.8, inner sep=0.4pt]
\tikzset{vertex/.style = {shape=circle,draw,minimum size=1.5em}}
\tikzset{edge/.style = {->,> = latex'}}
% vertices
\node[vertexY] (a) at  (0,3) {$i_{1}$};
%\node[vertexY] (b) at  (-1.0,1.5) {$i_{2}$};
\node[vertexY] (c) at  (0,0) {$i_{2}$};

\node[vertexY] (e) at  (1.5, 1.5) {$I$};

\node[vertexY] (l) at  (5, 1.5) {$J$};

\node[vertexY] (i) at  (6.5,0) {$j_{2}$};
%\node[vertexY] (j) at  (7.5, 1.5) {$j_{2}$};
\node[vertexY] (k) at  (6.5,3) {$j_{1}$};

%edges

\draw (a) edge (e);
%\draw[edge] (a) to (e);
%\draw[edge] (e) to (a);

%\draw (b) edge (e);
%\draw[edge] (b) to (e);
%\draw[edge] (e) to (b);

\draw (c) edge (e);
%\draw[edge] (c) to (e);
%\draw[edge] (e) to (c);

\draw (e) edge (l);
%\draw[edge] (e) to (l);
%\draw[edge] (l) to (e);

\draw (i) edge (l);
%\draw[edge] (i) to (l);
%\draw[edge] (l) to (i);

%\draw (j) edge (l);
%\draw[edge] (j) to (l);
%\draw[edge] (l) to (j);

\draw (k) edge (l);
%\draw[edge] (j) to (l);
%\draw[edge] (l) to (j);

\end{tikzpicture}
\caption{A six person social network.}
\label{fig:twoStarConnected}
\end{figure}

%\begin{figure}[ht!]
%\centering
%\includegraphics{Images/twoStarConnected.pdf}
%\caption{Two star graphs with central vertices connected}
%\label{fig:twoStarConnected}
%\end{figure}

As in Example~\ref{ex:Netflix}, each individual wishes to utilise the online media services provider Netflix.
The point of this example is to highlight how the most-efficient equilibria can vary non-monotonically with $\shareNumber$. 
Paradoxically, efficiency need not be gained by increasing the ability to share.
(To economise on space we are somewhat loose and describe equilibria by listing only the $D$-set, i.e., those that purchase in equilibrium.)

When $\shareNumber$ is equal to 1 for everybody, there is a unique equilibrium with $D$-set $D_{1} = \set{i_{1}, i_{2}, j_{1}, j_{2}}$.
%When $\shareNumber$ is equal to 2, $D_{1}$ is again the only $D$-set.
When $\shareNumber$ is increased to 2, three new equilibria emerge.
This collection of equilibria have $D$-sets given by $D_{1}$, $D_{2} = \set{I, J}$, $D_{3} = \set{i_{1}, i_{2}, J}$ and $D_{4} = \set{I, j_{1}, j_{2}}$.
Thus for $\shareNumber$ equal to 2, there is an equilibrium with $D$-set of size 2.
This equilibrium is depicted in Figure \ref{fig:twoStarConnectedEquilibria} below as a subgraph of the social network with the same colour coding as in Example~\ref{ex:Netflix}.
% (drivers in \blue{blue} and passengers in \red{red}, and arrows between the two originating at drivers making an offer).
When $\shareNumber$ increases to $3$, the model reduces to the best-shot game for which the $D$-set in equilibrium form a maximal independent set (the collection of which is $D_{1}, D_{3},$ and $D_{4}$).
Note that the most efficient equilibrium in this case has $D$-sets is of size 4, so this is less efficient than when $\shareNumber = 2$.

\begin{figure}[ht!]
\centering
\begin{tikzpicture}

\tikzstyle{vertexX}=[circle,draw, top color=gray!10, bottom color=gray!70, minimum size=10pt, scale=0.9, inner sep=0.5pt]
\tikzstyle{vertexY}=[circle,draw, top color=black!10, bottom color=red!70, minimum size=25pt, scale=0.8, inner sep=0.4pt]
\tikzstyle{vertexZ}=[circle,draw, top color=black!10, bottom color=blue!70, minimum size=25pt, scale=0.8, inner sep=0.4pt]
\tikzset{arc/.style = {->,> = latex'}}
\tikzset{vertex/.style = {shape=circle,draw,minimum size=1.5em}}
\tikzset{edge/.style = {->,> = latex'}}
% vertices
\node[vertexY] (a) at  (0,3) {$i_{1}$};
%\node[vertexY] (b) at  (-1.0,1.5) {$i_{2}$};
\node[vertexY] (c) at  (0,0) {$i_{2}$};

\node[vertexZ] (e) at  (1.5, 1.5) {$I$};

\node[vertexZ] (l) at  (5, 1.5) {$J$};

\node[vertexY] (i) at  (6.5,0) {$j_{2}$};
%\node[vertexY] (j) at  (7.5, 1.5) {$j_{2}$};
\node[vertexY] (k) at  (6.5,3) {$j_{1}$};

%edges

\draw [thick, dashed] (e) -- (l); 

%\draw (a) edge (e);
%\draw[edge] (a) to (e);
\draw [arc,line width=1.3pt] (e) -> (a); 
%\draw[-{Latex[length=3mm]}] (e) to (a);

%\draw (b) edge (e);
%\draw[edge] (b) to (e);
%\draw [arc,line width=1.3pt] (e) -> (b); 
%\draw[-{Latex[length=3mm]}] (e) to (b);

%\draw (c) edge (e);
%\draw[edge] (c) to (e);
\draw [arc,line width=1.3pt] (e) -> (c); 
%\draw[-{Latex[length=3mm]}] (e) to (c);

%\draw (e) edge (l);
%\draw[edge] (e) to (l);
%\draw[edge] (l) to (e);

%\draw (i) edge (l);
%\draw[edge] (i) to (l);
\draw [arc,line width=1.3pt] (l) -> (i); 
%\draw[-{Latex[length=3mm]}] (l) to (i);

%\draw (j) edge (l);
%\draw[edge] (j) to (l);
%\draw [arc,line width=1.3pt] (l) -> (j);
%\draw[-{Latex[length=3mm]}] (l) to (j);

%\draw (k) edge (l);
%\draw[edge] (k) to (l);
\draw [arc,line width=1.3pt] (l) -> (k);
%\draw[-{Latex[length=3mm]}] (l) to (k);

\end{tikzpicture}
\caption{Most efficient equilibrium occurs with $\shareNumber = 2$.}
\label{fig:twoStarConnectedEquilibria}
\end{figure}

%\begin{figure}[ht!]
%\centering
%\includegraphics{Images/twoStarConnectedEquilibria.pdf}
%\caption{Equilibrium with minimal number of adopters for $\shareNumber = 3$}
%\label{fig:twoStarConnectedEquilibria}
%\end{figure}

We now note some other interesting features of the model illustrated by this example.
First, when $\shareNumber$ is equal to 1, there is no equilibrium in which $I$ or $J$ are contained in a $D$-set.
In fact, when $\shareNumber$ is equal to 1, there is unique equilibrium.
The possibility of a player not providing in some equilibrium and the possibility of a unique equilibrium highlights a difference between our model and both the best-shot game and the model of BK for which equilibria correspond to maximal independent sets.\footnote{We note that however that there are other models of public good provision on networks that also admit unique equilibria (e.g., \cite{Allouch:2015:JET} and \cite{KinatederMerlino:2021:arXiv}).}
That is, two immediate differences between independent dominating sets and $D$-sets are the following:\ (i) every vertex is always part of at least one independent dominating set, and (ii) there are always at least two independent dominating sets (unless the graph being considered has no edges).

Second, and referencing the original conjecture, the equilibrium with the smallest $D$-set occurs with $\shareNumber$ equal to 2 for each individual, at which point individuals $I$ and $J$ both have degree greater than their capacity.\footnote{Note that the size of the smallest $D$-set can not only change with capacity but can do so to an arbitrary extent. To see this suppose in the example above that $I$ and $J$ each have $k$ friends. Then for $\shareNumber = k-1$, the smallest $D$-set is the set of the peripheral agents with size $2k -2$, for $\shareNumber = k$, the smallest $D$-set is $\set{I, J}$ of size 2, but for $\shareNumber = k+1$ the smallest $D$-set has size $k+1$.}
Furthermore, this is an equilibrium in which two neighbours in the social network both buy a Netflix account.

Third, the number of equilibria has changed from one when $\shareNumber = 1$, to four when $\shareNumber =2$, to three when sharing is maximal ($\shareNumber = 3$).
%Furthermore, while increasing capacity may lead to an increase in the number of equilibria, it may remove the most efficient equilibrium (that with the minimum $D$-set).

Fourth, and finally, consider an amendment to the social network in Figure \ref{fig:twoStarConnected} such that $i_{1}$ and $j_{2}$ are also connected.
For $\shareNumber = 1$ there remains an equilibrium with $D$-set $D_{1}$, but there is also an equilibrium with $D$-set $\set{i_1, i_2, j_2}$ in which $i_{1}$ offers access to $j_{1}$ and $i_{j}$ accepts.
We note that the latter is a Pareto improvement as these $D$-sets can be ordered by (strict) set-inclusion.
%
%Lastly, again supposing that $i_{1}$ and $i_{2}$ are friends, when everyone's capacity is maximal the largest $D$-set is the maximal independent set of maximum size.
%This has 5 elements.
%These two equilibria cannot be Pareto ranked since in the second case either $i_{1}$ or $i_{2}$ is no longer driving.
\end{example}

\begin{example}[Non-specialised equilibria are unstable]\label{ex:dynamics}
The previous examples were special cases of the general model - those where the action set is simply $\set{0, 1}$ (i.e., do not purchase / purchase).
In the general model, the action set is the set of non-negative integers and there is a quantity, denoted by $\qstar$, beyond which marginal cost exceeds marginal benefit.
Once an individual is consuming $\qstar$, she would gladly accept more of the good from neighbours who nominate her, but she would not be willing to pay for any excess herself.

In this richer model, there can be pure strategy equilibria where some individuals choose a strictly positive action choice that is less than $\qstar$.
Following the terminology of \cite{BramoulleKranton:2007:JET}, we call equilibria in which all individuals choose either action $0$ or action $\qstar$ \emph{specialised}.
(In the binary action Netflix Game of Examples~\ref{ex:Netflix} and \ref{ex:twoStars}, the action set is $\set{0, 1}$, so all pure strategy equilibria are specialised by definition.)
We now give an example of a pure strategy equilibrium that is not specialised and is not stable.
In Section~\ref{DYNAMICS}, Proposition \ref{proposition:necessary} shows that this is not an artefact of this example since all non-specialised equilibria are unstable.

We consider a social network with 6 individuals denoted $i, j, k, \ell, m$, and $n$. Individual $j$ is linked to everybody, individuals $i$ and $k$ are linked to $j$ and one other, while all other individuals, $\ell, m$, and $n$, are linked to $j$ and exactly two others. This network is depicted in Panel A of Figure \ref{fig:dynamics} below.

\begin{figure}[ht!]
\centering
\scalebox{0.7}
{\includegraphics{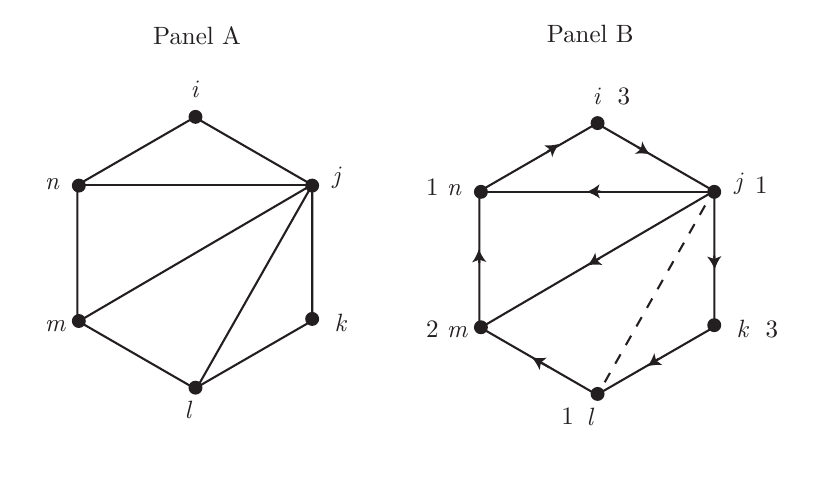}}
\caption{A 6-person network}
\label{fig:dynamics}
\end{figure}

We suppose that $\qstar = 4$.
In addition to the larger action space, this example is more complex in that we allow capacities to differ across agents. Specifically, we suppose that the capacity for $j$ is equal to 3 and the capacity for all other individuals is equal to 1.
%Lastly, we suppose that the optimal quantity of the good is $\qstar = 4$.

Panel B of Figure \ref{fig:dynamics} presents a pure strategy equilibrium that is not specialised. At this equilibrium individuals choose positive quantities where the quantity is given by the number beside their label.
As before arrows depict nominations with the number of arrows originating at any vertex equal to that individual's capacity.
It is easy to check that the above nominations and action choices constitute a pure strategy Nash equilibrium:\ for every individual, simply verify that their own quantity choice added to the in-flow of quantity choices from the neighbours who nominate them sums to 4 $(=\qstar)$.

%The edge $j \ell$ in Panel B is dashed since neither $j$ nor $\ell$ nominate each other. That is, this was an edge in the original graph but will not be in the equilibrium $DP$-subgraph. 

Now let us introduce dynamics in choice of action while holding fixed the nominations of everyone.
Let us imagine that individual $i$ unilaterally decreases his action choice from $3$ to $2$.
We refer to the time, $t$, at which this happens by $0$.
We label the action profile at this time by $\sbold^{(0)}$, where ordering the players as before we have that $\sbold^{(0)} = (\strategy{i}^{(0)}, \strategy{j}^{(0)}, \strategy{k}^{(0)}, \strategy{\ell}^{(0)}, \strategy{m}^{(0)}, \strategy{n}^{(0)}) = (2, 1, 3, 1, 2, 1).$ Our interest is in the evolution of the sequence of action profiles $\set{\sbold^{(t)}}_{t \geq 0}$ where for all $t \geq 1$, elements in the profile $\sbold^{(t)}$ specify the best-action response for the relevant individual to $\sbold^{(t-1)}$.

We reiterate that the nominations of each individual are held fixed so the only updating that is done is in action choice.
This is restrictive but to study dynamics without this simplification is not easy since without it the network could also be evolving every period.
%An example of where a problem could arise can be seen from considering a specialised equilibrium and considering the behaviour of an individual whose action choice was $0$.
We note that in the static case, it did not matter who an individual with action choice of $0$ nominated.
But in the dynamic environment, wherein this individual may change their action choice, such repercussions that can only be analysed and tracked if all nominations are known.\footnote{In Section~\ref{NomDYNAMICS} we consider unilateral deviations in the nomination of one player. After the change in the nomination, the nominations of everyone, including that of the deviator, are held fixed. Equilibria robust to this form of nomination are said to be {\it nomination stable}.}

In terms of how individuals update their action choice, we assume the behavioural rule of myopic best-response (in action, not in nomination as this is held fixed).
Specifically, each agent examines the total that he currently consumes.
If this total is less than 4, then he makes up the difference by increasing his own supply; if this total is more than 4 then he decreases his own supply (if possible).
Thus given the reduction in $i$'s action choice from 3 to 2 in period $t=0$, in period $t=1$, the only individuals who will alter their action are $i$ and $j$ as each receive a total of 3 ($i$ provides 2 himself and receives 1 from $n$ who nominated him, while $j$ provides 1 himself and receives 2 from agent $i$ who nominated him).
Since both are 1 unit short of $\qstar = 4$, they will each increase their period $t=0$ action by 1.
We thus get that $\sbold^{(1)} = (\strategy{i}^{(1)}, \strategy{j}^{(1)}, \strategy{k}^{(1)}, \strategy{\ell}^{(1)}, \strategy{m}^{(1)}, \strategy{n}^{(1)}) = (3, 2, 3, 1, 2, 1)$.
Given this, we can employ the same rules as before to compute $\sbold^{(2)}$, and so on.
Thus, the complete evolution of behaviour from action profile $\sbold^{(0)}$ can be tracked, and is presented in Table \ref{table:dynamics} below. 

% of th the rest to the reader to confirm. In period $0$, the only individuals for whom the quantity is not $\qstar = 4$ are $i$ and $j$ ($i$ provides 2 himself and receives 1 from $n$ who nominated him, while $j$ provides 1 himself and he receives 2 from agent $i$ who nominated him).We now examine the behaviour in future periods. It can be checked that the process evolves in a seemingly irregular way until period 13 at which point it starts to cycle. Table \ref{table:dynamics} below demonstrates this.

\begin{table}[htp!]
%[scale=0.6]
%\caption{Evolution of behaviour under best-action reply dynamic.}
\begin{center}
\begin{tabular}{|c||c|c|c|c|c|c|}
\hline
$t$& $\strategy{i}^{(t)}$ & $\strategy{j}^{(t)}$ & $\strategy{k}^{(t)}$ & $\strategy{\ell}^{(t)}$ & $\strategy{m}^{(t)}$ & $\strategy{n}^{(t)}$ \\
\hline\hline
1&2&1&3&1&2&1 \\
\hline2& 3&2&3&1&2&1\\
\hline3 &3&1&2&1&1&0\\
\hline4 &4&1&3&2&2&2\\
\hline5 &2&0&3&1&1&1\\
\hline6 &3&2&4&1&3&3\\
\hline7 &1&1&2&0&1&0 \\
\hline8 &4&3&3&2&3&2 \\
\hline9 &2&0&1&1&0&0\\
\hline10 &4&2&4&3&3&4 \\
\hline11 &0&0&2&0&0&0\\
\hline12 &4&4&4&2&4&4\\ 
\hline13 &0&0&0&0&0&0 \\
\hline14 &4&4&4&4&4&4\\
\hline15 &0&0&0&0&0&0\\
\hline
\end{tabular}
\caption{Evolution of behaviour under best-action reply dynamic.}
\end{center}
\label{table:dynamics}
\end{table}%

We note that the action profiles listed in Table~\ref{table:dynamics} evolve in a seemingly irregular way until period 13 at which point it starts to cycle between everybody choose $4$ in even periods and everybody choose $0$ in odd periods.
It is clear that a cycle of this form will never be exited and so any hope of returning to the equilibrium with action profile $\sbold^{(0)}$ is lost.

Let us summarise.
We began with the population at a non-specialised equilibrium.
We fixed the nominations and imposed a minimal change (1 unit) in the action choice of one individual.
From there, we let all players choose the optimal in every period going forward, and we noted that dynamics of this form led to a complete unravelling.
In Proposition \ref{proposition:necessary} we show that this is not an artefact of this particular example. 
Rather, for every non-specialised equilibrium, if the action choice of any agent not choosing $\qstar$ is perturbed slightly, then an unravelling of this sort is guaranteed to occur.

%action choice was minimally changed (by 1 unit) and then all players were allowed to update. And yet, this small deviation from the non-specialised equilibrium led to a complete unravelling. In Proposition \ref{proposition:necessary} we show that is not a coincidence. For every non-specialised equilibrium, if the action choice of any agent not choosing $\qstar$ is perturbed slightly, then an unravelling of this sort is guaranteed to occur.

\end{example}

%%% ----------------------------------------------------------------------
%%% ----------------------------------------------------------------------
%%% ----------------------------------------------------------------------
%%% ----------------------------------------------------------------------
%%% ----------------------------------------------------------------------
%%% ----------------------------------------------------------------------

%%% ----------------------------------------------------------------------
%%% ----------------------------------------------------------------------

%\newpage

\section{Public goods in networks with sharing constraints}\label{MODEL}

In Section~\ref{sec:ourModel} we lay out the model.
Section~\ref{sec:existence} begins with an existence theorem for a pure strategy Nash equilibrium in every instance of the model.
The proof requires some novel graph-theoretic concepts (e.g., $DP$-Nash subgraphs, $D$-sets, $P$-sets).
The section concludes with some observations on how these novel concepts relate to existing ones (e.g., dominating sets and independent dominating sets). 

\subsection{The model}\label{sec:ourModel}

%%% ----------------------------------------------------------------------
%%% ----------------------------------------------------------------------

We begin with the graph-theoretic terminology required to describe the model.

An \emph{undirected graph} $G = (V, E)$ consists of a nonempty finite set $V = V(G)$ of elements called \emph{vertices} and a finite set $E = E(G)$ of unordered pairs of distinct vertices called \emph{edges}. We call $V(G)$ the vertex set of $G$ and $E(G)$ the edge set of $G$. In other words, an edge $\set{i, j}$ is a 2-element subset of $V(G)$. We will often denote an edge $\set{i, j}$ by $ij$. For edge $ij \in E(G)$ we say that $i$ and $j$ are the \emph{end-vertices}, and say that end-vertices are \emph{adjacent}. We say that vertex $i$ is \emph{incident} to edge $e$ if it is an end-vertex of $e$. A graph $G$ on $n$ vertices is called {\em complete} if every two distinct vertices in $G$ are adjacent; $G$ will be denoted by $K_n.$

A \emph{path} in $G$ is a finite sequence of edges which connect a sequence of distinct vertices. A graph is \emph{connected} if there is at least one path containing each pair of vertices. We define the \emph{neighbourhood} of a vertex $i$, $\neighbourhood{G}(i)$, in a graph $G$ to be the set of vertices that vertex $i$ is adjacent to, $\neighbourhood{G}(i) = \set{j \in V \, : \, ij \in E}$, and we say that vertex $j \in \neighbourhood{G}(i)$ is a \emph{neighbour} of vertex $i$; we write $\degree{G}(i)$ for the cardinality of $\neighbourhood{G}(i)$. For a connected graph with at least two vertices, the neighbourhood of every vertex is nonempty.

%\red{While this is enough terminology to define the game theoretic model, there are more terms required for Theorems 1 and 2 (I am not sure what the terms required are if Gregory still plans to rewrite Theorem 1 and 2 in terms of dominating sets, etc.). The issue then becomes if we want to introduce all of that terminology now or wait until later when it is required.}

%\footnote{\red{I have not defined connected. Furthermore, there are some terms in the proofs of Theorems 1 and 2 that need to be defined concisely. These include: subgraph (and induced subgraph), bipartite graph, partite sets, what it means to add and delete edges, spanning subgraph, components, adjacent vertices, vertex incident to an edge.}}

With the above we can now introduce the game-theoretic component of the model. In the model, vertices are interpreted as players and edges represent connections between pairs of players. We assume the population (vertex set) is of size $n$ and that the graph $G$ is connected.

Let $X = \set{0, 1, \dots, \xbar}$ denote the finite set of actions common to each agent. Actions have both a private and (local) public benefit, but only a private cost.
{Writing $\Natural_{0}$ for the set of non-negative integers, there is a capacity function $\shareNumber: V \to \Natural_{0}$ that specifies, for each player, how many of his neighbours must \emph{be nominated} as benefactors of his action choice.}\footnote{An alternative set-up would be to require that each player $i$ must nominate no more than $\shareNumber(i)$ neighbours. As mentioned in the introduction we wish the limiting case of our model, that where $\shareNumber(i) = d_G(i)$ for all $i$, to be equivalent to the model of \cite{BramoulleKranton:2007:JET} and this alternative specification would then differ. In Section~\ref{NomDYNAMICS} we briefly consider this alternative set up, and we consider unilateral deviations in nomination. We show that for a large class of graphs, that includes many considered in BK, that only when individuals nominate up to their capacity are long-run stable.}
%Moreover, allowing individuals to nominate less neighbours than their capacity can lead to equilibria that seem to us implausible. See Appendix~\ref{APP:weakCapacity} and in particular Example~\ref{ex:whySaturate} and its associated discussion.
If a player's capacity is zero then he cannot nominate any neighbours.
If a player's capacity is at least as great as his degree then he is required to nominate all of his neighbours; that is, throughout when we write that each individual $i$ nominates $\shareNumber(i)$ neighbours we really mean that individual $i$ nominates $\min\set{\shareNumber(i), d_{G}(i)}$ neighbours - we will write $\shareNumber(i)$ to simplify notation.

Formally, for any nonempty set $A$ and nonnegative integer $k$, we denote by ${A \choose k}$ the collection of $k$-subsets of $A$.
That is, ${A \choose k} = \{A\}$ when $k \geq \abs{A}$, ${A \choose k} = \set{S \subseteq A \, : \, \abs{S} = k}$ when $0 < k < \abs{A}$, and ${A \choose k} = \emptyset$ when $k= 0$. With this, player $i$'s set of pure strategies is given by $\strategySet{} \times \nominateSet{i}{(\shareNumber)}$, with $\nominateSet{i}{(\shareNumber)} ={\neighbourhood{G}(i) \choose \shareNumber(i)}$ representing the collection of subsets of $\neighbourhood{G}(i)$ of size $\shareNumber(i)$. The size of each set $\nominateSet{i}{(\shareNumber)}$ is ${\degree{G}(i) \choose \shareNumber(i)}$.

%We emphasise the distinction between $\strategy{i} = 0$, meaning player $i$ not adopt, and $\strategy{i} = \emptyset_{i}$, meaning player $i$ adopts but nominates no neighbour (because $\shareNumber(i) = 0$).

We write $\strategy{i}$ for individual $i$'s action choice from $\strategySet{}$, and $\nominate{i}$ for his nominating choice from $\nominateSet{i}{}$. (Note we omit the superscript $(\shareNumber)$ of $\nominateSet{i}{(\shareNumber)}$ whenever no confusion arises.) A pure strategy profile is represented by a vector $(\sbold, \nominateBold) = \big((\strategy{1}, \dots, \strategy{n}), (\nominate{1}, \dots, \nominate{n})\big)$ specifying an action and a set of nominees for each agent. (We call $\sbold$ the \emph{action profile} and $\nominateBold$ the \emph{nomination profile}.)\footnote{While we have assumed that the graph $G$ is exogenously given, the model permits an alternative interpretation. Specifically, suppose initially that there is no network so that nobody is connected:\ each individual $i$ is required to nominate \underline{any} $\nominate{i}$ others to whom he wants to link, and the resulting nominations induce the formation of a social network. See Appendix~\ref{APP:endogenous}.}
%For any nonempty set $A \subseteq V$ we write $\sbold_{A} = \set{\strategy{j}}_{j \in A}$, and we write $\sbold_{-i}$ when $A = V \backslash \set{i}$. Likewise for $\nominateBold_{A}$. If player $i$ nominates player $j$ we write $j \in \nominate{i}$.
%We will often abbreviate this and simply say that $i$ nominates $j$.

We define the utility to player $i$, $\utility{i}$, from strategy profile $(\sbold, \nominateBold)$, as
\begin{equation}\label{eq:utility}
\utility{i} \big( \sbold, \nominateBold \big) = f\Big(\strategy{i} + \sum_{\{ j \in \neighbourhood{G}{(i)} \,\, : \,\, i \in \nominate{j} \}} \strategy{j} \Big) - c \strategy{i}
\end{equation}
where we assume that (i) $c > 0$, and (ii) there exists a $\qstar \in X$ such that 
\begin{equation}\label{eq:assumption}
\qstar \in \argmax_{x \in X} \left( f(x) - cx \right) \hspace{.2in} \text{and} \hspace{.2in} f(x) - cx \text{ is non-increasing for } x \geq \qstar.
\end{equation}

While mentioned before, we repeat that each player $i$'s utility solely depends upon his own action choice and the action choices of his neighbours in $G$ who nominate him. Player $i$'s utility does not depend upon who he himself nominates. 

The assumption given in \eqref{eq:assumption} implies that it makes no sense for a player to increase their provision if they are already receiving $\qstar$ in total (i.e., if $\strategy{i} + \sum_{\{ j \in \neighbourhood{G}{(i)}\,\, : \,\, i \in \nominate{j} \}} \strategy{j} =y\geq \qstar$).
This is the case because for any such $y \ge \qstar$, it must be that $ f(y+t)-c(y+t)\leq f(y)-cy$ and hence $f(y+t)-ct \leq f(y)$, a fact that will be used in the proof of Theorem~\ref{theorem:balancedPSNE}.
In words this says that while an individual will always (at least weakly) benefit from more of the good when they are receiving $\qstar$, the marginal benefit of any increase beyond $\qstar$ is exceeded by the marginal cost.

%q∗, at which an individual would always (at least weakly) benefit from more of the good but would never pay for more themselves, since the marginal benefit of each additional unit is nonincreasing whereas marginal cost is constant.

%Note that this definition implies that it makes no sense for a player $i$ to increase $\strategy{i}$ if $\strategy{i} + \sum_{\{ j \in \neighbourhood{G}{(i)}\,\, : \,\, i \in \nominate{j} \}} \strategy{j} =y\geq \qstar$ because $ f(y+t)-c(y+t)\leq f(y)-cy$ and thus $f(y+t)-ct \leq f(y)$.
%Thus, each player $i$ chooses an action $\strategy{i} \in \strategySet{i}$ and decides which neighbours to share with via the choice $\nominate{i} \in \nominateSet{i}{}$. Note that player $i$'s utility solely depends upon his own action choice and the action choices of his neighbours in $G$ who nominate him. Player $i$'s utility does not depend upon who he himself nominates. 

The utility function as defined in \eqref{eq:utility} is general. While it does not require the concavity of $f$ as in that of \cite{BramoulleKranton:2007:JET}, this comes at the cost of requiring the action set to be discrete. If we take $X = \set{0, 1}$, $f(x) = 1$ for all $x \geq 1$, and $0 < c < 1$, then the game becomes the ``Netflix Game with $\shareNumber$-user sharing rule''.

%We do not explicitly model payoffs. Rather we consider the class of games with best-responses of the following form.\footnote{\label{footnote:payoffs}One class of payoff functions that generate the above come from the following story. Access to the technology brings a payoff of $1$, whereas not having access yields $0$.  For each player $i$, adopting the technology has an associated cost of $c(i) \in (0, 1)$. For this class of game, each player most prefers when a neighbour adopts and nominates him. The next best outcome is to adopt (and nominate $\min{\set{\shareNumber(i), d_{G}(i)}}$ neighbours). The worst outcome is not adopting and not being nominated by a neighbour who adopts.}

%For player $i$ we define $\bestResponse{i}{\strategyProfile}$, his best-response to strategy profile $\strategyProfile$, as follows,
%\[
%\bestResponse{i}{\strategyProfile} = \left\lbrace
%  \begin{array}{l l}
%    0, & \text{ if } i \in \strategy{j} \text{ for some } j \in \neighbourhood{i},\\
%    \text{some } \strategy{i} \in \subsetsSizeShareNumber{\neighbourhood{}(i)}{\shareNumber(i)}, & \text{ otherwise }\\
%  \end{array}
%\right.
%\]

\subsection{Existence of pure strategy Nash equilibrium}\label{sec:existence}

Our interest is in pure strategy Nash equilibria.
These are defined in the usual way:

\begin{definition}
A strategy profile $(\sboldstar, \nominateBoldStar)$ is a pure strategy Nash equilibrium if for every $i = 1, \dots, n$, and every $\strategy{i} \in \strategySet{i}$ and every $\nominate{i} \in \nominateSet{i}{}$ we have
\[ 
\utility{i}\left( (\sboldstar, \nominateBoldStar)\right) \geq \utility{i} \left((x^*_1,\ldots,x^*_{i-1},x_i ,x^*_{i+1},\ldots,x^*_n)(m^*_1,\ldots,m^*_{i-1},m_i ,m^*_{i+1},\ldots,m^*_n)\right).
\]
%\[
%\utility{i}\big( (\strategy{i}^{*}, \nominate{i}^{*}), (\sboldstar_{-i}, \nominateBoldStar_{-i}) \big) \geq \utility{i}\big((\strategy{i}, \nominate{i}), (\sboldstar_{-i}, \nominateBoldStar_{-i}) \big)
%\]

\end{definition}

We have the following existence result:

\begin{proposition}\label{theorem:PSNE}
A pure strategy Nash equilibrium exists for any graph $G$ and any capacity function $\shareNumber: V \to \Natural_{0}$.\footnote{Our model does not admit a potential function \citep{ShapleyMonderer:1996:GEB} which would render the existence of a pure strategy equilibrium immediate.}
\end{proposition}

We will see that Proposition~\ref{theorem:PSNE} follows from Theorem~\ref{theorem:balancedPSNE} that is proved below.	

While the above is a strong result, we will focus on what, following \cite{BramoulleKranton:2007:JET},  we term \emph{specialised strategy profiles} - those in which each agent either choose action $\strategy{i} = 0$ or $\strategy{i} = \qstar$. We have the following definition.

\begin{definition}
A specialised strategy profile is a pure strategy profile in which for all $i \in V$ we have either $\strategy{i} = 0$ or $\strategy{i} = \qstar$.
\end{definition}

For a given specialised strategy profile $(\sbold, \nominateBold)$, let $D(\sbold, \nominateBold) = \set{i \in V : \strategy{i} = \qstar}$ be those individuals who supply $\qstar$, and $P(\sbold, \nominateBold) = \set{i \in V : \strategy{i} = 0}$ be those who supply nothing.
(In equilibrium those in $D$ provide while those in $P$ free ride.)
Clearly, at any specialised strategy profile, we have that both $D \cap P = \emptyset$ and $D \cup P = V$.

Not all specialised strategy profiles support equilibrium.
There are two reasons.
First of all, everyone choosing action $0$ is a specialised strategy profile but clearly not an equilibrium.
Second, there is another subtle complication with specialised profiles that does not occur in the model of BK as in their model individuals nominate all neighbours.
That is, even a specialised strategy profile in which neighbouring individuals choose consistent quantities need not be an equilibrium as the nominations must also ``match up'' correctly.
To see this consider the equilibrium depicted in Figure \ref{fig:twoStarConnectedEquilibria}.
If action choices stayed the same but $I$ nominated $J$ instead of her friend $i_{1}$, then population behaviour is no longer an equilibrium since $i_1$ no longer has access.
So for a specialised profile to be an equilibrium we need consistency in the nominations.
In words, we wish to find a nomination profile such that nobody in $D$ is nominated by somebody else in $D$, and everyone in $P$ is nominated by at least on person in $D$.

%
%\marginpar{\tiny{\PN{PN: The AE flagged these sentences as related to ``balanced''. Depending on your responses hopefully easy to deal with.}}}
%\red{However, we aim to go further. In words, we wish to find a strategy profile such that nobody in $D$ is nominated by somebody else in $D$, and everyone in $P$ is nominated by at least on person in $D$.}

%\red{Note to all: There is some graph theoretic terminology (e.g. spanning, bipartite, partite sets, subgraph, adding)  that is used in subsequent definitions and Theorem \ref{theorem:balancedPSNE} that I introduce below. The choice is whether to introduce it all at the beginning or pause now and reintroduce it. Philip slightly prefers introducing the relevant definitions at the relevant time, but his preferences are not strong. Philip also thinks he may have butchered the below! Probably best for Gregory and Stefanie to take a close look at the following paragraph and state more concisely.}

To describe the sort of consistency in nominations that we require, some additional graph-theoretic terminology is required.
Given a graph $G = (V(G), E(G))$, we say that $H = (V(H), E(H)) $ is a \emph{subgraph} of $G$ if $V(H) \subseteq V(G)$, $E(H) \subseteq E(G)$, {and every edge in $E(H)$ has both end-vertices in $V(H)$}.
If $V(H) = V(G)$ we say that $H$ is a \emph{spanning subgraph} of $G$.
%\marginpar{\tiny{\PN{Notation for subgraphs needed. See comment from AE.}}}
A subgraph $H$ is said to be {\em induced} by a subset $S$ of vertices of $G = (V, E)$, if the vertex set of $H$ is $S$ and the edge set consists of all edges in $E$ that have both end-vertices in $S$. If $G = (V, E)$ is a graph and $S\subseteq V(G)$, we write $G- S$ for the subgraph induced by $V(G)\setminus S$.
%For a subgraph $H$ of $G$, we define $G - H$ as $G - V(H)$.
%Similarly, for edges, if $B \subseteq E(G)$, then $G-B$ is the spanning subgraph of $G$ with edge set $E(G) - B$.

Bipartite graphs play a particularly important role in our analysis.
A {\em bipartite graph} is a graph whose vertices can be partitioned into two disjoint sets (called {\em partite sets}) $A$ and $B$ such that every edge has one end-vertex in $A$ and the other in $B$.
A bipartite graph $G$ with partite sets $A$ and $B$ is called {\em complete bipartite} if $ab\in E(G)$ for every $a\in A$ and $b\in B$.
Then $G$ is denoted by $K_{|A|,|B|},$ where $|A|$ and $|B|$ are the cardinalities of the sets $A$ and $B.$
A complete bipartite graph $K_{1,p}$ ($p\ge 2$) is called a {\em star}, the vertex adjacent to all other vertices {\em the center}, all other vertices {\em leaves}.
For example the graph in Figure~\ref{fig:5star} is a star with centre $\ell$.

Abstracting for now from the actions chosen by the individuals, we wish to find a spanning bipartite subgraph, $H$, of $G$, with partite sets $D$ and $P$ such that there does not exist an $i \in D$ such that $i \in \nominate{j}$ for any $j \in D$, and for all $i \in P$ there exists at least one $j \in D$ such that $i \in \nominate{j}$. This leads us to the following definition.

\begin{definition}\label{def:DPNash}
A spanning bipartite subgraph $H$ of $G$ with partite sets $P$ and $D$ is called a {\em $DP$-Nash subgraph} if for each $i \in D$ the degree of $i $ in $H$ is $\min\set{\shareNumber(i),\degree{G}(i)}$ and for every $i\in P$ the degree of $i$ in $H$ is positive.
\end{definition}

%With this, we can define what it means for a strategy profile to be a \emph{balanced specialised profile} as follows.
%
%\begin{definition}\label{def:balanced}
%A specialised profile $(\sbold, \nominateBold)$ is a \emph{balanced specialised profile} supported by $DP$-Nash subgraph $H$ of $G$ if
%\[
%		\begin{cases}
%			\strategy{i} = \qstar, \,\, \nominate{i} = \neighbourhood{H}(i) \,\,  & \text{ if } i  \in D \\
%			\strategy{i} = 0, \,\, \nominate{i} = S \text{ for some } S \in  \nominateSet{i}{(\shareNumber)} \,\,  & \text{ if } i  \in P
%		\end{cases}
%\]
%\end{definition}

%\red{\text{ such that }\,\, A \cap \neighbourhood{H}(i) \neq \emptyset}

We will see that $DP$-Nash subgraphs precisely characterise the nomination component of strategy profile that generate the consistency in nominations that support specialised pure strategy Nash equilibria.

We have the following theorem, which is an improvement of Proposition \ref{theorem:PSNE}.

\begin{theorem}\label{theorem:balancedPSNE}
For any graph $G$ and common utility function as given by \eqref{eq:utility}, there exists a specialised pure strategy Nash equilibrium.
\end{theorem}

%\red{Maybe the above could be rephrased exactly as Theorem 0}

\begin{proof}\label{proof:theoremBalancedPSNE}
The proof has two parts. The first is to show that every graph $G$ possesses at least one {\em $DP$-Nash subgraph}. The second is then to show that a specialised profile induced by $DP$-Nash subgraph involves each agent choosing the optimal action. 
%We show the first part here and relegate the second part to Appendix \ref{App:Proofs} as it is really just a marshalling of definitions.

Here we prove the following claim:\ If $G$ is a graph and $\shareNumber:\ V(G) \rightarrow \Natural_0$ is a function, then $G$ has a $DP$-Nash subgraph.

The proof proceeds by induction on the number $n$ of vertices of $G.$ If $n=1$, then $G$ is a DP-Nash subgraph with  $D=V(G)$ and $P=\emptyset$.

Now assume the claim is true for all graphs with fewer than $n_{0}\geq 2$ vertices, and let $G$ be a graph on $n_{0}$ vertices.

\begin{description}
\item[ Case 1: There is a vertex $i$ of degree at most $\shareNumber(i).$] 

Let $B$ be  the star with center $i$ and leaves $N_{G}(i),$ where $N_{G}(i)$ is the neighbourhood of $i$ in $G$. Let $G'=G-V(B)$. Set $D=\{i\}$ and $P=N_{G}(i).$ If $G'$ has no vertices then $B$ is clearly a $DP$-Nash subgraph of $G.$

Otherwise, by induction hypothesis, $G'$ has a $DP$-Nash subgraph $H'$ with partite sets $P'$ and $D'$. Construct a subgraph $H$ of $G$ from the disjoint union of $H'$ and $B$ by adding to it for every $j\in D'$ with $\degree{G'}(j)<\shareNumber(j),$ exactly $\min\{\degree{G}(j),\shareNumber(j)\}-\degree{G'}(j)$ edges of $G$ between $j$ and $N(i).$ Set $D=D'\cup \{i\} $ and $P=P'\cup N(i).$

To see that $H$ is a $DP$-Nash subgraph of $G,$ observe that (a) $H$ is a spanning bipartite subgraph of $G$ as $H'$ and $B$ are bipartite and the added edges are between $D$ and $P$ only, (b) every vertex $j\in D$ has degree in $H$ equal to $\min\{\shareNumber(j), \degree{G}(j)\},$ (c) every vertex $k \in P$ is of positive degree (since it is so in both $B$ and $H'$). 

\item[ Case 2: For every vertex $j\in V(G)$, $\degree{G}(j)>\shareNumber(j)$.]  Let $i$ be an arbitrary vertex. Delete $\degree{G}(i)-\shareNumber(i)$ edges incident to $i$ and denote the resulting graph by $L$. Observe that every $DP$-Nash subgraph of $L$ is a $DP$-Nash subgraph of $G$ since no vertex in $L$ has degree less than $\shareNumber(i).$ This reduces Case 2 to Case 1.
\end{description}

To complete the proof it remains to show that a specialised profile induced by a $DP$-Nash subgraph is a Nash equilibrium.
Let  $(\sboldstar, \nominateBoldStar)$ be a such a specialised profile.
As mentioned before, we note that the utility function defined in \eqref{eq:utility} does not depend on $\nominate{i}$. Thus player $i$ cannot increase her payoff by deviating from $\nominate{i}^{*}$. If $i \in D$ then by the definition  of  a nicely specialised profile induced by a $DP$-Nash subgraph, we have $\strategy{i}^{*} = \qstar$ and $ \sum_{\{ j \in \neighbourhood{G}(i) : \,\, i \in \nominate{j}^{*} \}} \strategy{j}^{*}=0$. Thus is follows from the definition of the utility of $i$ that for all $\sbold$ that are obtained from $\sboldstar$ by replacing $\sboldstar_i=\qstar$ by any $t>0$, we have 
\[ U_i(\sboldstar,\nominateBoldStar)=f(\qstar)-c\qstar\geq f(t)-ct=U_i(\sbold,\nominateBoldStar).\]

If $i\in P$ then $\strategy{i}^{*} = 0$ and $ \sum_{\{ j \in \neighbourhood{G}(i) : \,\, i \in \nominate{j}^{*} \}} \strategy{j}^{*}=s \qstar$ for some $s\geq1$. 
%Note that for any $t>0$ by (ii) in the definition of the utility function, we have  
%$ f(s \qstar)-c (s\qstar)\geq f(s\qstar + t)-c(s\qstar+t)$ and thus $f(s\qstar)\geq f(s\qstar+t)-ct$.
Thus by the observation just after the definition of the utility function,  for all $\sbold$ that are obtained from $\sboldstar$ by replacing $\sboldstar_i=0$ by any $t>0$ we calculate
\[ U_i(\sboldstar,\nominateBoldStar)=f(s\qstar)\geq f(s\qstar+t)-ct=U_i(\sbold,\nominateBoldStar).\]
\end{proof}

We now make some observations about Theorem \ref{theorem:balancedPSNE} and more generally about $DP$-Nash subgraphs and their associated $D$-sets and $P$-sets.

Fix a graph $G = (V, E)$.
A \emph{dominating set} is a set of vertices such that every vertex in $V$ is either in the set or has a neighbour in the set, and a \emph{minimal dominating set} is a dominating set that does not contain a proper subset that is dominating.
(Note that the notion of dominating is defined only for sets whereas our concept of nominating is defined for individual vertices and by considering multiple nominating vertices can be extended to sets.)
An \emph{independent set} is a subset of vertices no pair of which are adjacent, and a \emph{maximal independent set} is an independent set that is not a proper subset of any other.
A maximal independent set must be a dominating set and so is also referred to as an independent dominating set.
We now relate $D$-sets to independent dominating sets and minimal dominating sets.
For a further discussion of some of the properties of $D$-sets, refer to Appendix \ref{App:DSets}.

Clearly a $D$-set is a dominating set though the reverse need {not} hold. An example of this is a complete graph on $n$ vertices, $K_n$, wherein every vertex forms a dominating set. But if $\kappa(i)<d(i)$ for every $i\in V(K_n),$ then no singleton can be a $D$-set.

We add further observations.
First, as with dominating sets but not independent sets, it is possible that two vertices in a given $D$-set are adjacent in $G$.
An instance of this was seen in Example \ref{ex:twoStars} of Section \ref{EXAMPLE} for the equilibrium in which the two central vertices made up the $D$-set.
Second, when $\shareNumber(i) \geq \degree{G}(i)$ for all $i$, we have that a $D$-set is an independent dominating set. 
Third, and related to the previous observation, is that unlike independent dominating sets, for a given graph $G$ and capacity function $\shareNumber$, one $D$-set may be a strict subset of another.
When this occurs, for the binary action Netflix Game it is possible to Pareto-rank the two corresponding specialised equilibria (since the equilibrium supported by the smaller $D$-set has a strict subset of the individuals purchasing). 
As an example consider a complete graph with 5 vertices $i, j, k, l$, and $m$, with $\shareNumber = 2$ for each vertex. 
One such $D$-set is $\set{i, j, k}$ with each nominating $l$ and $m$, while another $D$-set is $\set{i, j}$ with $i$ nominating $k$ and $l$ and $j$ nominating $l$ and $m$.
Clearly $\set{i, j} \subset \set{i, j, k}$.

%Finally we note that the procedure described in the proof of Theorem \ref{theorem:balancedPSNE} allows us to find a $D$-set in polynomial time.\footnote{For a classification of what Netflix games admit a unique pure strategy equilibrium see [citation]. It is also shown that the question of uniqueness raises interesting issues of computational complexity. In particular, it is proven that (i) whether an instance of the model has a unique $DP$-Nash subgraph can be decided in polynomial time, and (ii) when the capacity of every individual is equal to some integer $k$, the problem of deciding whether a Netflix game has a unique $D$-set is polynomial time solvable for $k=0$ and 1, but is {\sf co-NP}-complete for $k\ge2.$} Note, however, that such a procedure may not find all $D$-sets of a given graph G. For example the $D$-set given in Example \ref{ex:twoStars} consisting of $I$ and $J$ would never be found.

%First we note that the procedure described in the proof of Theorem \ref{theorem:balancedPSNE} allows us to find a $D$-set in polynomial time.\footnote{Note, however, that such a procedure may not find all $D$-sets of a given graph G. For example the $D$-set given in Example \ref{ex:twoStars} consisting of $I$ and $J$ would never be found.}

We conclude this section by highlighting two issues related to computational complexity.\footnote{For a comprehensive presentation of computational complexity see \cite{GareyJohnson}.}
First, given its constructive nature, the proof of Theorem~\ref{theorem:balancedPSNE} suggests an algorithm that will always find a specialised equilibrium in polynomial time. 
Note, however, that there are instances of the model for which a procedure based on the proof of Theorem~\ref{theorem:balancedPSNE} will not find all specialised equilibria.
For example, the specialised equilibrium of Example~\ref{ex:twoStars} at which $\shareNumber(I) = \shareNumber(J) = 3$, that is depicted in Figure~\ref{fig:twoStarConnectedEquilibria}, would not be found by the procedure.

Second, unlike the best-shot game and the BK model of public goods, instances of our more general model can possess a unique specialised equilibrium.
Example \ref{ex:twoStars} in Section \ref{EXAMPLE} provided such instances when $\shareNumber$ is equal to 1 or 2.
This is because a graph can have a unique $D$-set but always has at least two maximal independent sets.
The issue of uniqueness generates a surprising dichotomy in computational complexity between specialised equilibrium and specialised equilibrium outcomes.
This is in part because specialised equilibria are characterised by $DP$-Nash subgraphs whereas specialised equilibrium outcomes are characterised by $D$-sets.
\cite{GutinNeary:2023:DAM} show that the problem {\sc $DP$-Nash Subgraph Uniqueness}:\ {\it decide whether a capacitated graph has a unique $DP$-Nash subgraph}, can be decided in polynomial time.\footnote{The problem of finding a minimum dominating set weakly satisfying the capacity constraints exists already, and is referred to as the {\sc Capacitated Domination Problem} \citep{Guha:2003vc,CyganPW11,KaoCL15}. Capacitated domination is itself a special (discrete) variant of the well-studied {\sc Facility Location Problem} (see \href{https://en.wikipedia.org/wiki/Facility_location_problem}{https://en.wikipedia.org/wiki/Facility\_location\_problem}). So our model generates a new graph-theoretic problem that can be viewed as the {\sc Exact Capacitated Domination Problem}.}
However, the nearby problem {\sc $D$-set Uniqueness}:\ {\it decide whether a capacitated graph has a unique $D$-set}, is {\sf co-NP}-complete.
{That is, deciding whether an instance of our model has a unique specialised equilibrium can be checked quickly, but deciding whether an instance of our model has a unique specialised equilibrium outcome cannot.}

%\footnote{The problem of finding a minimum dominating set satisfying this weak inequality exists already in graph theory, and is referred to as the {\sc Capacitated Domination Problem} \citep{Guha:2003vc,CyganPW11,KaoCL15}. Capacitated domination is itself a special (discrete) variant of the well-studied {\sc Facility Location Problem} (see \href{https://en.wikipedia.org/wiki/Facility_location_problem}{https://en.wikipedia.org/wiki/Facility\_location\_problem}). So our model generates a new graph-theoretic problem that can be viewed as the {\sc Exact Capacitated Domination Problem}.}

%%% ----------------------------------------------------------------------
%%% ----------------------------------------------------------------------
%%% ----------------------------------------------------------------------
%%% ----------------------------------------------------------------------
%%% ----------------------------------------------------------------------
%%% ----------------------------------------------------------------------

%%% ----------------------------------------------------------------------
%%% ----------------------------------------------------------------------
%%% ----------------------------------------------------------------------
%%% ----------------------------------------------------------------------
%%% ----------------------------------------------------------------------
%%% ----------------------------------------------------------------------

%%% ----------------------------------------------------------------------
%%% ----------------------------------------------------------------------

%\newpage

\section{Efficiency and comparative statics}\label{COMPARATIVESTATICS}

%%% ----------------------------------------------------------------------
%%% ----------------------------------------------------------------------

In this section we focus on the efficiency of specialised equilibria.
%For a given graph and a given capacity function $\shareNumber,$ t
There are often multiple specialised profiles and so there may be interest in which equilibrium is optimal according to some metric.
If, for example the function $f(\cdot)$ in equation~\eqref{eq:utility} is strictly increasing, it is immediate that all specialised equilibria are Pareto efficient.
But even in such cases there are other measures of efficiency that can be considered.
The measure of efficiency/inefficiency that we will focus on are those with the smallest and largest $D$-sets - certainly these are the specialised equilibria that are least/most costly.\footnote{An alternative measure would be to adopt a utilitarian social welfare criterion - we will touch on this briefly at the end of this section.}
%We adopt the utilitarian social welfare criterion as our measure of efficiency, so our attention is naturally on those equilibria with the smallest and largest $D$-sets.

Our approach will be to consider how incremental amendments to the model affect the set of (in)efficient specialised equilibria.
One natural way to do this is to incrementally increase the capacities of the individuals.

%, while the second is to alter the underlying graph, $G$, either by adding/deleting an edge or by adding/deleting a vertex and all the edges it is incident to.

%In this section we focus on the efficiency of balanced specialised profiles. For a given graph and a given capacity function $\shareNumber,$ there are often multiple balanced strategy profiles. We adopt the utilitarian social welfare criterion as our measure of efficiency, so our attention is on those equilibria with the smallest and largest $D$-sets. We then consider how incremental amendments to the model will affect such (in)efficiency. There are two natural ways to do this. The first is to incrementally increase the capacities of the agents, while the second is to alter the underlying graph, $G$, either by adding/deleting an edge or by adding/deleting a vertex and all the edges it is incident to.

Recall that for a specialised equilibrium profile $(\sboldstar, \nominateBoldStar)$, we write $D(\sboldstar, \nominateBoldStar)$ to denote the set of individuals who adopt.
We have the following definition.

\begin{definition}\label{def:efficient}
A specialised equilibrium $(\sboldstar, \nominateBoldStar)$ is said to be {\it efficient} ({\it inefficient}) if its associated $D$-set, $D(\sboldstar, \nominateBoldStar)$, is of minimal (maximal) size.
\end{definition}

For a graph $G$ and capacity function $\shareNumber : V \to  \Natural_{0}$, let $\delta_{\min}^{\shareNumber}(G)$ and $\delta_{\max}^{\shareNumber}(G)$ denote the minimum and maximum sizes of $D$-sets of $G$.
Our focus is then on how, for a given graph $G$, the two sets $\delta_{\min}^{\shareNumber}(G)$ and $\delta_{\max}^{\shareNumber}(G)$ vary as the sharing capacity increases.
To this end, let $\shareNumber':\ V(G) \rightarrow \mathbb{N}_0$ be a function such that $\shareNumber(i)\le \shareNumber'(i)$ for every $i\in V(G)$.
We compare $\delta^{\shareNumber}_{\min}(G)$ and $\delta^{\shareNumber}_{\max}(G)$ with $\delta^{\shareNumber'}_{\min}(G)$ and $\delta^{\shareNumber'}_{\max}(G)$.

Theorem \ref{thm:delta} shows a particular inequality holds for every graph $G$.\footnote{It can be shown that none of the other three possible inequalities hold. We will provide counterexamples to each in the discussion following the theorem.}
In words the theorem says the following.
Fix a graph, $G$, and fix two capacity functions $\shareNumber$ and $\shareNumber'$ such that every vertex has at least as much capacity under $\shareNumber'$ as under $\shareNumber$. 
Then, the smallest $D$-set for $\shareNumber'$ is bigger than the largest $D$-set for $\shareNumber$.
That is, the most efficient equilibrium for $\shareNumber'$ is always more efficient than the least efficient equilibrium under $\shareNumber$.

%\footnote{\red{Related to this:\ Consider a graph $G$ and a function $\shareNumber$. Can anything be said about when it is it the case that $\delta^{\shareNumber}_{\max}(G)$ is weakly efficient? For example, it is not for the complete graph unless $\shareNumber = n-1$.}} 

\begin{theorem}\label{thm:delta}
For every graph $G$, and any two capacity functions $\kappa$ and $\kappa'$ such that $\kappa(i) \leq \kappa'(i)$ for every $i \in V(G)$, we have that $\delta^{\shareNumber'}_{\min}(G)\le \delta^\shareNumber_{\max}(G)$.
\end{theorem}

\begin{proof}
Let $\shareNumber^+: V(G)\rightarrow \Natural_0$ be function such that $\shareNumber^+(i)=\shareNumber(i)$ for $i\in V(G)\setminus\set{j}$ and $\shareNumber^+(j)=\shareNumber(j)+1$ for some $j\in V(G).$
To prove the theorem is sufficient to show that $\delta^{\shareNumber^+}_{\min}(G)\le \delta^\shareNumber_{\max}(G)$.

We proceed by induction on  $n+m,$ where $n$ is the number of vertices of $G$ and $m$ is the number of edges in $G$.  If $n+m=1$, then $G$ consists of a single vertex and setting $D=V(G)$ and $P=\emptyset$ gives the {only} $DP$-Nash subgraph {for both $\shareNumber$ and $\shareNumber^+$}. We may assume that $G$ is connected as otherwise we can consider its components and apply the  induction hypothesis on the component containing $j$ and the vertices in the other components have the same values for $\shareNumber$ and $\shareNumber^+$.  Let $G=K_{1,n-1},$ where $n\ge 2$ and $j$ is the center of the star. If $\shareNumber^+(j)\ge n-1$, then $\delta^{\shareNumber^+}_{\min}(G)=1$ and we are done. Otherwise, $V(G)\setminus\set{j}$ is a $D$-set for both $\shareNumber$ and $\shareNumber^+$. 

Now we may assume that $n\ge 3$, $G$ is connected and there is an edge in $G$ which is not incident to {$j$}. Consider two cases. 

\begin{description}
\item[ Case 1: There is a vertex $i\in V(G)\setminus\set{j}$ of degree at most $\shareNumber(i).$] 
Let $B$ be the star with center $i$ and leaves $N_{G}(i),$ where
$N_{G}(i)$ is the neighbourhood of $i$ in $G$. Let $G'=G-V(B)$. Set $D=\{i\}$ and $P=N(i).$  If $G'$ has no vertices then $B$ is clearly a $DP$-Nash subgraph of $G$ for $\kappa$, and a $DP$-Nash subgraph of $G$ for $\shareNumber^+$.

Otherwise, by induction hypothesis, $\delta^{k^+}_{\min}(G')\le \delta^k_{\max}(G')$, where $k$ is $\shareNumber$ restricted to $G'.$ The two corresponding $DP$-subgraphs of $G'$ can be extended to those of $G$ by adding {$i$} to their $D$-sets and $N(i)$ to their $P$-sets 
and adding to every $\ell\in D$ with $d_{G'}(\ell)<\shareNumber(\ell),$ and  $d_{G'}(\ell)<\shareNumber^+(\ell),$ respectively, exactly $\min\{d_G(\ell),\shareNumber(\ell)\}-d_{G'}(\ell)$ or exactly $\min\{d_G(\ell),\shareNumber^+(\ell)\}-d_{G'}(\ell)$ edges of $G$ between $\ell$ and $N(i).$ Thus, 
$$\delta^{\shareNumber^+}_{\min}(G)\le \delta^{k^+}_{\min}(G')+1\le \delta^{k}_{\max}(G')+1\le \delta^{\shareNumber}_{\max}(G).$$

%\vspace{-1cm}

\item[ Case 2: The degree of every vertex $i\in V(G)\setminus\set{j}$ is larger than $\shareNumber(i).$] 
%If $x$ is an isolated vertex, then a $\ka$-$DP$ subgraph of $G-x$ can be extended to $\ka$-$DP$-subgraph and $\ka^+$-$DP$-subgraph of $G$ by adding $x$ to the $D$-set.
%Otherwise, 
Choose any edge $i\ell$ such that $j\not\in \{i,\ell\}$ and delete this edge from $G$.
By the induction hypothesis, the resulting graph $G'$ has $\delta^{k^+}_{\min}(G')\le \delta^k_{\max}(G')$.
It remains to observe that the two $DP$-Nash subgraphs of $G'$ are also $DP$-Nash subgraphs of $G$ (the functions $\shareNumber$ and $\shareNumber^+$ were not changed and the deleted edges are not needed). 
\end{description}
\vspace{-1cm}
\end{proof}

%In words Theorem \ref{thm:delta} says the following. Fix a graph, $G$, and fix two capacity functions $\shareNumber$ and $\shareNumber^{+}$ such that every vertex has at least as much capacity under $\shareNumber^{+}$ as under $\shareNumber$. Then, the most efficient equilibrium for $\shareNumber^{+}$ is never more inefficient than the least efficient equilibrium under $\shareNumber$.

It is easy to find counter examples to the other three potential inequalities. To see that $\delta_{\max}^{\shareNumber^+}(G)$ and $\delta_{\max}^\shareNumber(G)$ cannot be ordered we give examples of the inequalities in both directions. Consider a star $K_{1,5}$ and add an edge between any two non-central vertices, say $i$ and $j$, and let $\shareNumber(\ell)=5$ for the center $\ell$, and all other vertices $h$ have $\shareNumber(h)=1$.
Then the maximum $D$-set has size $5$. But if $\shareNumber(i)$ is increased then the maximum $D$-set has size $4$. On the other hand, consider the complete graph $K_7$ on $7$ vertices. Pick two vertices $i$ and $j$ and let $\shareNumber(i)=\shareNumber(j)=2$ and $\shareNumber(h)=6$ for all other vertices $h$. Then all $DP$-Nash graphs have $D$-sets of size $1$ (and all vertices except $i$ and $j$ can form $D$) and thus $\delta_{\max}^\shareNumber(G)=1$. But if we increase $\shareNumber(i)$ by one to obtain $\shareNumber^+$ then $\set{i, j}$ form a $D$-set and $\delta_{\max}^{\shareNumber^+}(G)=2$.

If we are interested in comparing $\delta_{\min}^{\shareNumber^+}(G)$ and $\delta_{\min}^\shareNumber(G)$ then we find a similar situation.  Consider the graph in Example \ref{ex:twoStars}. When $\shareNumber(i)=2$ for all vertices then we have seen that the minimum $D$-set has size $6$, but if we increase the $\shareNumber(I)$ and $\shareNumber(J)$ to 3, then the new minimum $D$-set has size $2$. If we further increase $\shareNumber(I)$ (or $\shareNumber(J)$) then the new minimum $D$-set has size $4$.

We conclude this section with brief discussion of an alternative measure of efficiency.
While the measure of efficiency we have employed is that of ``smallest $D$-set'', one might also wish to consider other welfare criterions like, for example, the utilitarian social welfare criterion.
We will now show that such a welfare criterion depends on the precise formulation of the payoff functions.
In particular, it depends on the structure of $f(\cdot)$ and value of $c$ in equation~\eqref{eq:utility}.

To see why, consider a society organised as per the network of Example~\ref{ex:twoStars}.
Consider the specialised equilibrium where both $I$ and $J$ each contribute $\qstar$ (depicted in Figure~\ref{fig:twoStarConnectedEquilibria}) and the specialised equilibrium where everyone but $I$ and $J$ contributes $\qstar$. 
In the first equilibrium the sum of utilities is $8f(\qstar) - 2c\qstar$, whereas in the second equilibrium the sum is $6f(\qstar) + 2f(3\qstar) - 6c\qstar$.
Which of these two sums is larger will depend on the precise structure of $f(\cdot)$ and the size of $c$.

%If we are interested in comparing $\delta_{\min}^{\shareNumber^+}(G)$ and $\delta_{\min}^\shareNumber(G)$ then we find a similar situation.  Consider the graph in Example \ref{ex:twoStars}. When $\shareNumber(i)=2$ for all vertices then we have seen that the minimum $D$-set has size $6$, but if we increaset if we increase the $\shareNumber$ for $I$ (or $J$) then the new minimum $D$-set has size $4$. 

%As a more striking example, consider the complete graph with $\kappa + 2$ vertices, $G=K_{\shareNumber+2}$ with $\shareNumber\ge 0$. Observe that $\delta^\shareNumber_{\min}(G)=2>\delta^{\shareNumber+1}_{\max}(G)=1$ for every $\shareNumber\ge 0.$

%Let us remark that adding a vertex can increase or decrease the size of a maximum $D$-set $D$. To see that it can increase just take the star and add another non-center vertex (or just add an isolated vertex to any graph). But it can also decrease. Take $K_6$ with $2$ vertices $i$ and $j$ with $\shareNumber(i)=\shareNumber(j)=2$ and all the other vertices have large $\shareNumber$. Then $i, j$ is a $D$-set. But if we add a vertex $h$ adjacent to every other vertex and let $\shareNumber(h)=6$, then all $D$-sets have size $1$.

%%% ----------------------------------------------------------------------
%%% ----------------------------------------------------------------------
%%% ----------------------------------------------------------------------
%%% ----------------------------------------------------------------------

%\newpage

\section{Stability of specialised profiles}\label{DYNAMICS}

%%% ----------------------------------------------------------------------
%%% ----------------------------------------------------------------------

%In this section we introduce dynamics with the goal of examining which strategy profiles are robust to unilateral deviations.
%{While non-specialised equilibria can be interpreted as equilibria wherein there is ``some cooperation / coordination \dots amongst the individuals of a given society'', we will see that such equilibria are, in a sense to be made precise soon, quite unstable.}

In this section we consider dynamics.
Our goal is to examine which strategy profiles are robust to unilateral deviations.
In our model individuals make a two-pronged decision - a choice of action and a choice of nomination.
We consider unilateral deviations in both.

Studying dynamics is complicated due to the richness of our model.
In particular, since each individual's utility is unaffected by their choice of nomination and the number of ways they can nominate can be enormous, some simplifications must be made.\footnote{Recall that the number of nominations for individual $i$ with capacity $\shareNumber(i)$ is ${\degree{G}(i) \choose \shareNumber(i)}$.}
%As such, we consider some simplifications.
We begin, in Section~\ref{ActDYNAMICS}, by assuming that the nomination profile is fixed. We impose a unilateral change in one person's action choice. From there, with nominations held fixed, we ``let the system go'' assuming best-action reply dynamics.
% in response to a deviation in action of some player.
In Section~\ref{NomDYNAMICS}, we suppose a deviation in the nomination of some player so that a new nomination profile is reached. That new nomination profile is then held fixed and the best-action reply dynamics are assumed.\footnote{Formally this sort of deviation considers a different model wherein each player need not nominate exactly $\shareNumber(i)$ neighbours, but rather {\it no more than} $\shareNumber(i)$ neighbours. More details in Section~\ref{NomDYNAMICS}.}

%Both variants involve considering a deviation and then ``letting the system go''.

One take-away message is shared by both sorts of deviations.
{While non-specialised equilibria can be interpreted as equilibria wherein there is some cooperation / coordination amongst the individuals of a given society, equilibria of this form are not stable.}
That is, specialised equilibria appear the only serious candidates for long run behaviour.

% consider only what action choices are optimal given the nominating profile and the action choices of others.

%Recalling that each individual's utility is unaffected by their choice of nomination, the number of best-responses may be enormous.\footnote{Recall that the number of nominations for individual $i$ with capacity $\shareNumber(i)$ is ${\degree{G}(i) \choose \shareNumber(i)}$.}
%As such, we will assume that the nominating profile is held fixed and consider only what action choices are optimal given the nominating profile and the action choices of others.

%We wish to be clear that there is no clear game-theoretic defence of this assumption.
%Rather we adopt this approach as it seems to us (i) the simplest way to begin an analysis of dynamics (and therefore a sensible place to start), and (ii) consistent with the applied intuition that links are slower to adapt than behaviour.

\subsection{Action-stability of specialised profiles}\label{ActDYNAMICS}

With respect to this dynamic based on \emph{best-action responses}, Proposition \ref{proposition:necessary} below shows that only specialised equilibria can be stable.
Theorem \ref{theorem:sufficient} below shows that specialised equilibria are stable only when every individual in $P$ is nominated by at least two individuals in $D$.
This last result is the analog of Theorem 2 in \cite{BramoulleKranton:2007:JET}), and we emphasise that it does not hold for all networks.
The reason being that there exist networks that do not possess a $DP$-Nash subgraph satisfying the density condition.

Given a nomination profile $\nominateBold$ and the action profile $\sbold$, it is not hard to see that the best-action response of agent $i$, $\bestResponse{i, \nominateBold}{\sbold}$, is given by
\begin{equation}\label{eqn:bestResponse}
\bestResponse{i, \nominateBold}{\sbold} = \max \set{ \qstar - \sum_{\set{j \in \neighbourhood{i}(G) \, : \, i \in \nominate{j}}} \strategy{j}, \,\,\, 0}
\end{equation}
We extend this to the best-action reply dynamic $\Bcal_{\nominateBold} : \Sbold \to \Sbold$ as
\begin{equation}\label{eqn:bestReply}
\bestReply{\nominateBold}{\sbold} = \Big(\bestResponse{1, \nominateBold}{\sbold}, \bestResponse{2, \nominateBold}{\sbold},  \dots, \bestResponse{n, \nominateBold}{\sbold} \Big)
\end{equation}

We assume time, $t$, is discrete and starts at $t= 0$. The interpretation will be that we set the action profile to $\sbold$ at time $t=0$ and then `let the system go' with it evolving according to \eqref{eqn:bestReply} above.

\begin{definition}\label{def:brEvolution}
Given an action profile $\sbold$ we define the \emph{best-action evolution of $\sbold$} recursively by 
$\sbold^{(0)}=\sbold$, and for all $t\geq 1$, $\sbold^{(t)} = \bestReply{m}{\sbold^{t-1}}$.
\end{definition}

Lemma \ref{lemma:nashSequencePopulation} in Appendix \ref{App:Lemmas} says the following. Consider a pure strategy Nash equilibrium, $(\sboldstar, \nominateBoldStar)$. Now, fix the nomination component, $\nominateBoldStar$, and weakly lower everyone's action choice so that behaviour changes from action profile $\sboldstar$ to some new action profile $\sbold \leq \sboldstar$. Then, the best-action evolution of $\sbold$ either oscillates around $\sboldstar$ forever, or there exists an $t_0$ such that for all $t\geq t_0$,  $\sbold^{(t)}=\sboldstar$.\footnote{We emphasise that Lemma \ref{lemma:nashSequencePopulation} holds for all Nash equilibrium strategy profiles and not simply specialised ones.} This motivates the following definition.

%Note that this is the case as we consider only finite networks and so there are only a finite number of different possibilities for the $\sbold^{t}$ and therefore they have to repeat. 

\begin{definition}
We say that the best-action evolution of $\sbold$ {\it settles} in $\sboldstar$ if there exists an $t_0$ such that for all $t\geq t_0$,  we have $\sbold^{(t)}=\sboldstar$. 
\end{definition}

We now introduce our notion of stability when the nominating profile is held fixed and actions are updated according to the best-action reply dynamic. In words, we say that action profile $\sbold$ is stable relative to the nomination profile $\nominateBold$, if, when the action of any individual is changed by some strictly positive amount, repeated application of the best-action reply dynamic will lead population behaviour back to action profile $\sbold$.
%\marginpar{\tiny we could just choose $\delta =1$ as this is all we need in the proofs}
\begin{definition}\label{definition:stability}
We say that strategy profile $(\sbold, \nominateBold)$ is \emph{action stable} if there exists $\delta \geq 1$ such that for any individual $i=1,\dots, n$ and any action profile $\sbold'$  with $x_j=x'_j$ for $j\not= i$ and  $|\strategy{i}'-\strategy{i}|\leq \delta$
the best action evolution of $\sbold'$ settles in $\sbold$.
\end{definition}
In words, a strategy profile is action stable if one can change any one action coordinate by at most $\delta$ and the best-action evolution will settle in $\sbold$. 
Proposition \ref{proposition:necessary} below shows that any pure strategy Nash equilibrium profile that is not specialised is not action stable and thus an equilibrium being specialised is necessary for stability.

\begin{proposition}\label{proposition:necessary}
Suppose that $(\sboldstar, \nominateBoldStar)$ is a pure strategy Nash equilibrium such that $0 < \strategy{\ell}^{*} < \qstar$ for some $\ell$. Then $(\sboldstar, \nominateBoldStar)$ is not action-stable.
\end{proposition}

Proposition \ref{proposition:necessary} says that specialised profiles supported by $DP$-Nash subgraphs are necessary for stability.
However, we will now show that such specialised equilibria are not sufficient for stability.
The issue with such specialised equilibria is that they are agnostic on who those in $P$ nominate.
We will show that specialised equilibria in which additional constraints are placed on who those in $P$ nominate are sufficient for stability.

We begin by supposing that a specialised profile $(\sbold, \nominateBold)$ supported by $DP$-Nash subgraph $H$ of $G$ for every individual $i \in P$ we have $\nominate{i} \cap \neighbourhood{H}(i) \neq \emptyset$.
That is, we require that those in $P$ nominate at least one individual who nominates them.
One might conjecture that such a mirroring in nominations would ensure stability.
But Proposition \ref{proposition:notSufficient} below, the proof of which is found in Appendix~\ref{App:Lemmas}, shows that the additional property described above is still not sufficient for stability. 
The reason is that the specialised equilibrium in which $\nominate{i} \cap \neighbourhood{H}(i) \neq \emptyset$ for all $i \in P$ may still allow that some individual in $P$ is nominated by only one individual from $D$, and this is enough to mean that the specialised profile is not action stable.

%However, Proposition \ref{proposition:sufficient} shows that balanced specialised profiles supported by a $DP$-Nash subgraphs with a certain properties sufficient for stability.
%
%By Proposition \ref{proposition:notSufficient}, if there exists a nicely balanced specialised profile with an individual in $P$ who is nominated by only one agent in $D$, then that equilibrium is not action-stable. However, the following proposition shows that if either all individuals in $P$ nominate only individuals in $P$  or all individuals in $P$ are nominated by at least two individuals in $D$ ($2 \geq \frac{q^*}{q^*-1}$ for $q^* > 1$), then the equilibrium is action-stable.

\begin{proposition}\label{proposition:notSufficient}
Suppose $(\sboldstar, \nominateBoldStar)$ is a specialised Nash equilibrium induced by $DP$-Nash subgraph $H$ of $G$, where (i) $\nominate{i} \cap \neighbourhood{H}(i) \neq \emptyset$ for all $i \in P$, and (ii) $\degree{H}(i) = 1$ for some $i \in P$.
Then $(\sboldstar, \nominateBoldStar)$ is not action-stable.
\end{proposition}

While Proposition~\ref{proposition:notSufficient} is negative, Theorem \ref{theorem:sufficient} shows that if all individuals in $P$ nominate only individuals in $P$ then the specialised equilibrium is action stable.
Moreover, if all individuals in $P$ are nominated by at least two individuals from $D$ (since $2 \geq \frac{q^*}{q^*-1}$ for all $q^* > 1$), then the specialised equilibrium is action-stable.

\begin{theorem}\label{theorem:sufficient}
	Suppose that $(\sboldstar, \nominateBoldStar)$ is a specialized profile induced by $DP$-Nash subgraph $H$ of $G$ and suppose that either 
	\[
      (i) \quad		m_i^* \cap N_H(i) = \emptyset  \text{ for all }i \in P  \text{ or (ii) }\quad  \min_{ i \in P}  |N_H(i)| \geq \frac{q^*}{q^*-1} 
	\]
holds. Then $(\sboldstar, \nominateBoldStar)$ is action-stable.
\end{theorem}

%\footnote{\red{Why not just set $\varepsilon = 1$ in the definition of stability?}}

\begin{proof}
Suppose that condition (i) holds.  Since $H$ bipartite, if $x_i^*=q^*$, then $m_i \subset P$. Thus for $i \in D$, we have
\[
	\text{ for all } l \in P, \,\, m_l \not \ni i, \qquad \text{ for all } l \neq i \in D, \,\, m_l=N_H(l) \not \ni i
\]
Thus for all $l  \neq i\in G$, $m_l \not \ni i$ and hence no vertex in $D$ is nominated by any neighbour. Suppose that for $i \in D$, we have
$x^{(0)} = (x_i^*\pm 1, x_{-i}^*)$ or for $i \in P$, $x^{(0)} = (x_i^*+ 1, x_{-i}^*)$. Then for all $i \in D$, $x^{(t)}_{i}=q^*$ for $t \geq 1$. Given this, for all $i \in P$, $x^{(t)}_{i}=0$ for all $t\geq 2$. Thus  $(\sboldstar, \nominateBoldStar)$ is stable.

Suppose that condition (ii) holds. Choose $\varepsilon = 1$. We must consider deviations by those in $D$ and those in $P$. We begin with those in $D$.
\begin{description}
\item[Case 1: $i \in D.$] By definition $\strategy{i}^{*} = \qstar$. There are two subcases: Either we increase $x_j^*$ by one or we decrease $x_j^*$ by one. 
In the first case, i.e.\,  $\sbold$ satisfies $x_i^{(0)}=x^*_i+1=\qstar+1$ and for all $j\not=i$,  $x_j^{(0)}=x_j^*$,  
it is immediate that $\bestResponse{\nominateBoldStar}{\sbold}=\sboldstar$. 

In the second case let $\sbold$ satisfy $x_i^{(0)}=x^*_i-1=\qstar-1$ and for all $j\not=i$,  $x_j^{(0)}=x_j^*$. By our assumption for all $j\in P$ there is at least one neighbour $k\not=i$, $k\in D$ with $j\in m^*_k$ and thus  $\max \set{ \qstar - \sum_{\set{k \in \neighbourhood{j}(G) \, : \, j \in \nominate{k}}} \strategy{k}, \,\,\, 0}=0$. It follows that $x_j^{(1)}=0$ for all $j\in P$. For all $j\in D$ we have $x_j^{(1)}=\qstar$ by Lemma~\ref{lemma:nashSequencePopulation} (or one can just observe only vertices in $k\in P$ are nominating $j$ and all these vertices satisfy $x_k^{(0)}=0$).

\item[Case 2: $i \in P.$] By definition $\strategy{i}^{*} = 0$, so the only deviation is to choose $\strategy{i}^{(0)} = 1$ and for all $j\not=i$,  $x_j^{(0)}=x_j^*$. Then
\[
\bestResponse{\ell, \nominateBoldStar}{\sbold^{(0)}} = \left\lbrace
  \begin{array}{l l}
    0, & 0 \text{ if } \ell \in P,\\
    \qstar - 1 > 0, & \ell \in \neighbourhood{H}(i) \cap D,\\
    \qstar, & \ell \in D \setminus \neighbourhood{H}(i)
  \end{array}
\right.
\]
Consider an individual $h \in P$ with $h \neq i$, and consider how all of $H$'s neighbours in the nominating network will behave in period 1. We have
\begin{align*}
\sum_{\set{j \in \neighbourhood{G}(h) \, : \, h \in \nominate{j}}} \strategy{j}^{(1)} &= \sum_{\set{j \in \neighbourhood{G}(h) \, : \, h \in \nominate{j}}} \bestReply{j, \nominateBoldStar}{\sbold^{(0)}}\\
&\geq \degree{H}(h) (\qstar - 1)\\
& \geq 2(\qstar - 1)\\
&\geq \qstar
\end{align*}
where the last inequality follows since $\qstar \geq 2$. And thus, $\bestReply{h, \nominateBoldStar}{\sbold^{(0)}} = 0$.\\
Furthermore, Since $H$ is a $DP$-Nash subgraph, we have that $\neighbourhood{H}(i) \subseteq P$ for all $i \in D$, and thus the set $\set{j \in \neighbourhood{G}(i) \, : \, i \in  \nominate{j}} \subseteq P$ for all $i$ in $D$. Thus, for all $\ell \in D$, we must have that $\Bcal^{2}_{\ell, \nominateBoldStar}(\sbold^{(0)}) = \qstar$ since $\bestReply{i, \nominateBoldStar}{\sbold^{(0)}} = 0$ for all $i \in P$.\\
Therefore, $(\sboldstar, \nominateBoldStar)$ is stable. 
\end{description}
\vspace{-1cm}
\end{proof}

Theorem \ref{theorem:sufficient} provides two sufficient conditions for strategy profile $(\sboldstar, \nominateBoldStar)$ to be action-stable.
The first of these requires that every non-specialist nominates only fellow non-specialists. Such profiles are at one extreme end of the nominating spectrum in that no individual in $P$ may nominate any individual in $D$ who nominated them.
To give an example of such an equilibrium that is action-stable, consider a three person environment where the social network is given by a complete graph, and two individuals have capacity equal to 1 and the remaining individual has maximal capacity equal to 2.
Let us consider the pure strategy equilibrium where the maximal capacity individual is in $D$, and they nominate both other individuals; the remaining two individuals comprise $P$, with each of them using their one nomination for the other.
This equilibrium is depicted in Figure \ref{fig:3complete} below, where the vertex shaded black vertex comprises the $D$-set and the unfilled vertices make up the $P$-set, and the direction of the arrows indicates who nominated who.

\begin{figure}[ht!]
\centering
\begin{tikzpicture}

\tikzset{vertex/.style = {shape=circle,draw,minimum size=1.5em}}
\tikzset{edge/.style = {->,> = latex'}}
% vertices
\node[vertex] (a) at  (0,3) {};
\node[vertex] (b) at  (0,0) {};
%\node[vertex] (c) at  (3,0) {$j$};
%\node[vertex] (d) at  (3,3) {$i$};

\node[fill = black, vertex] (e) at  (3, 1.5) {$\ell$};

%edges
\draw[-{Latex[length=3mm]}] (e) to (a);

\draw[-{Latex[length=3mm]}] (a) to (b);
\draw[-{Latex[length=3mm]}] (b) to (a);

\draw[-{Latex[length=3mm]}] (e) to (b);
%\draw (c) edge (e);

%\draw (e) edge (d);

%\draw[edge] (a)  to[bend left] (a1);
%\draw[edge] (a1) to[bend left] (a);

%\draw[edge] (a1) to[bend left] (a2);
%\draw[edge] (a2) to[bend left] (a1);
%
%\path (a2) to node {\dots} (c);
%\node [shape=circle,minimum size=1.5em] (a3) at (4.5,0) {};
%\draw[edge] (a2) to[bend left] (a3);
%\draw[edge] (a3) to[bend left] (a2);

%\node [shape=circle,minimum size=1.5em] (c1) at (6.5,0) {};
%\draw[edge] (c) to[bend left] (c1);
%\draw[edge] (c1) to[bend left] (c);
\end{tikzpicture}
\caption{Equilibrium on 3-person network}
\label{fig:3complete}
\end{figure}

To see that this equilibrium is action-stable, we need to consider two kinds of deviations.
The first is that wherein the sole player in $D$ changes his action (so he reduces his provision), and the second that wherein one of the those in $P$ changes their action (so they start providing). Figure \ref{fig:3completeDynamics} below shows repeated application of the best-action reply dynamic for each of these starting points. In both cases population behaviour will ultimately settle on the original equilibrium.

\begin{figure}[ht!]
\centering
\includegraphics{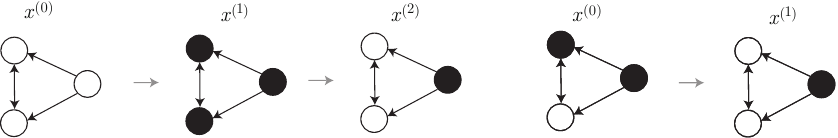}
\caption{Illustration that the equilibrium in Figure \ref{fig:3complete} is action-stable}
\label{fig:3completeDynamics}
\end{figure}

%The denominator, $q^*-1$, compares the scale of the required quantity of local public good, $q^*$, to the size of a single individual's deviation. 
Now consider condition (ii) in Theorem \ref{theorem:sufficient}.
This places a minimum requirement on the number of those in $D$ that each player in $P$ must be nominated by.
And given that the left hand side of the expression is integer-valued, each player in $P$ must be nominated by at least two players in $D$. We note that being nominated by at least two players in $D$ is precisely the condition for stability in \cite{BramoulleKranton:2007:JET}.\footnote{BK consider a continuous action set so that small deviations are given by $\varepsilon$. With this, the expression reduces to $\min_{ i \in P}  |N_H(i)| \geq \frac{q^*}{q^*-\epsilon}$. The smallest positive integer satisfying this is 2.}

%However, if we consider deviations in action that exceed 1, then the requirement on the number of neighbours each individual in $P$ must be nominated by increases.

%If we considered deviations in action that exceed 1, then the requirement on the number of neighbours each individual in $P$ must be nominated by increases. For example, if $\qstar = 3$ and we consider a deviation of 2, then the degree of each passenger in $H$, the $DP$ Nash subgraph that supports the equilibrium, must be 3. And so on. Viewed in this way, the stable equilibria in our setup can be regarded as having a relatively \textit{small} basin of attraction.

Thus, Theorem \ref{theorem:sufficient} partially generalises BK's result, by identifying conditions (on individual nominations) for stability in the more constrained environment where maximal sharing is not allowed.
However, we emphasise that, just as in BK, there are networks for which neither of the conditions of Theorem \ref{theorem:sufficient} can ever be satisfied in equilibrium.

%world where sharing is constrained results hold. 
%
%the analog of the Effectively, the condition requires that the size of a single driver's deviation is \textit{small} relatively to the amount of the local public good that he will provide.Viewed in this way, the stable equilibria in our setup can be regarded as having a relatively \textit{small} basin of attraction. {If we had considered a continuous action space, then condition (ii) in Proposition \ref{proposition:sufficient} becomes
%\[
% \min_{ i \in P}  |N_H(i)| \geq \frac{q^*}{q^*-\epsilon}  \text{ for small } \epsilon \iff \min_{ i \in P}  |N_H(i)| \geq 2
% \]
%which is precisely the condition for stability in \cite{BramoulleKranton:2007:JET}. That is, when capacity is maximal every individual in $P$ must have at least two neighbours who are in $D$. Thus we have partially generalised BK's result by identifying conditions (on individual nominations) under which their results hold. Finally, we emphasise that, just as in BK, there are networks for which neither of the conditions of Proposition \ref{proposition:sufficient} can ever be satisfied.

We conclude this section on action-stability a brief discussion of the recent paper by \cite{BervoetsFaure:2019:JET} who further explore the stability of equilibria in the BK model.
%Consider the specialised action profile in Figure \ref{fig:3complete}, though suppose that the environment is one in which capacity is maximal for one and all.
Consider the environment of Figure \ref{fig:3complete}, though suppose instead that the sharing capacity is maximal for one and all.
While we have depicted one specialised equilibrium there are two others - simply rotate which player comprises the singleton $D$-set.
(We note however that neither of these two specialised equilibria are stable in our environment.)
In a continuous action environment where $\qstar = 1$, \cite{BervoetsFaure:2019:JET} point out that the three specialised equilibria are extreme points of a connected component of Nash equilibria $\set{(x_i^*, x_j^*, x_k^*) | \sum_l x^{*}_{l} =1 \text{ and } x^{*}_{l} \geq 0, \text{ for all } l}$.
\cite{BervoetsFaure:2019:JET} then show that if stability is extended from a point-valued notion to a set-valued one, then every instance of the BK model contains a stable set of equilibria of the form above.\footnote{Extending the point valued solution concept of \cite{Nash:1951:AM,Nash:1950:PNASUSA} to set valued variants is common even in static environments due to limitations that arise when defining solutions by points. See for example the \emph{strategically stable sets} of \cite{KohlbergMertens:1986:E}.}

In our setup with a discrete action space, every Nash equilibrium is isolated and there does not exist a non-trivial connected component of Nash equilibria (except for singleton sets) and this is why the technical machinery of \cite{BervoetsFaure:2019:JET} cannot be extended.

\subsection{Nomination-stability of specialised profiles}\label{NomDYNAMICS}

%%% ----------------------------------------------------------------------
%%% ----------------------------------------------------------------------

In the previous section we considered unilateral deviations in action choice.
We now consider unilateral deviations in nomination choice.

Given the potentially enormous number of possible deviations in nomination, recall the number is ${\degree{}(i) \choose \shareNumber(i)}$, we consider an alternative model in which each individual $i$ must nominate no more than $\shareNumber(i)$ neighbours.
That is, the nomination sets are now replaced with $\nominateSet{i}{\scriptscriptstyle{\leq}}$, where $\nominateSet{i}{\scriptscriptstyle{\leq}} := \set{ S \subseteq \neighbourhood{G}(i): |S| \leq \shareNumber(i) }$, and everything else in the model remains as before (i.e., common action set given by $X = \set{0, 1, \dots, \xbar}$, utilities as given in equation~\eqref{eq:utility}, etc.).

%\marginpar{\tiny{\PN{Need to qualify. some maybe not feasible.}}}
While the number of total nominations may now be far greater for some players, we can consider deviations of the kind not viable before. Specifically, we allow each individual to increase or decrease the number nominated neighbours by 1.
That is, for a given nomination profile $\nominateBold$, we denote by $\nominateBold+ (i \rightarrow j)$ (or $\nominateBold- (i \rightarrow j)$) the new nomination profile formed by adding (removing) a nomination of $j$ by $i$ (where possible).\footnote{We qualify with ``where possible'' because if $i$ is already nominating $j$, then $\nominateBold+ (i \rightarrow j)$ is not a deviation; if $i$ is already nominating $\shareNumber(i)$ neighbours and $j$ is not one of them, then again $\nominateBold+ (i \rightarrow j)$ is not an allowable deviation; etc.}
Given nomination profile, $\nominateBold$, we write $\set{\nominateBold \pm (i  \rightarrow j)}$ for the collection of feasible deviations from nomination it.

%Our notion of stability will 

For a given strategy profile $(\sbold, \nominateBold)$, we consider a feasible deviations in nomination, and we then suppose that the new nomination profile is fixed.
From there we consider the best-action evolution (see Definition~\ref{def:brEvolution}) from the new profile. 
Our notion of nomination stability is then defined as follows.

\begin{definition}
Strategy profile $(\sbold, \nominateBold)$ is {\it nomination stable} if for every feasible deviation $\nominateBold' \in \set{\nominateBold \pm (i  \rightarrow j)}$, the best-action evolution of $(\sbold, \nominateBold')$ settles in $\sbold$.
\end{definition}

%Let $M_i := \{ S \subset N_G(i): |S| \leq \kappa(i) \}$ and $\mathbf{M}= \prod_i M_i$.
%For a given nomination profile $\nominateBold$, we denote by $\nominateBold+ (i \rightarrow j)$ (or $\nominateBold- (i \rightarrow j)$) a new nomination profile formed by adding (removing) a nomination of $j$ by $i$. We say that $\nominateBold \pm (i  \rightarrow j) \in \mathbf{M}$ is a feasible deviation from $\nominateBold$.  
%\begin{definition}
%We say that $(\textbf{x},\textbf{m})$ is nomination stable if for all $i$, for all $j$ such that $\nominateBold \pm (i \rightarrow j) \in \mathbf M $, $\sbold':= \mathcal{B}_{\nominateBold \pm (i \rightarrow j)} (\sbold)$ settles in $\sbold$.
%\end{definition}

%\noindent

%\marginpar{\tiny{\PN{``remains invariant'' or ``ultimately returns'' to?}}}
In words, we say that $(\sbold, \nominateBold)$ is nomination stable if for every feasible deviation from $\nominateBold$, the action profile $\sbold$ is the long run outcome.
That is, $\sbold$ is invariant in the long-run under the best-action response dynamic.

% (see Figure \ref{fig:nom-stab}).}

%
%\begin{proposition}
%Suppose that $G$ is a ring network. Suppose that $(x, m)$ is a specialized profile and nomination stable. Then 
%\[
%	|m_i| = 2 \text{ for all } i \in D 
%\]
%\end{proposition}
%
%\begin{proof}
%Let  $(x, m)$ be a specialized profile and nomination stable.
%Suppose that for some $i \in D$, $|m_i| < 2$. Then, there exist $j \in N_G(i)$ such that $ i \not \rightarrow j$. 
%
%\noindent \textbf{Case 1:}  $j \in D$.
%Then let $x' = B_{m + (i \rightarrow j)}(x)$. Then, $j \in D$ and $i \in D$ in $m+(i \rightarrow j)$ cannot constitute a NE, $x'$ cannot settle in $x$.
%
%\noindent\textbf{Case 2:} $j \in P$.
%Then $(x,m)$ is a NE, there exists $k \neq i$ and $k \in N_G(j)$ and $k \in D$.  Let $x'=B_{m - (k \rightarrow j)}(x)$. Then, $|N_H(j)|=1$ (because of the ring network),  $x' \neq x$ is a NE. Thus $x'$ does not settle in $x$.
%\end{proof}
%
%\newpage

A {\it clique} in a graph $G = (V, E)$ is a subset of vertices such that every pair of vertices in the subset are adjacent.
We now introduce a class of networks, \biClique s, for which nomination stability of a specialised equilibrium for any graph in the class involves the full saturation of capacity constraints for the specialists.
These graphs have the property that the neighbourhood of every vertex can be partitioned into two disjoint subsets each of which is a clique.

\begin{definition}\label{def:biclique}
%\marginpar{\tiny{\PN{We need a name??}}}
Say that connected graph $G = (V,E)$ is a {\it \biClique} if for every vertex $i$, the neighbourhood of $i$, $\neighbourhood{}(i)$, can be partitioned into two cliques $\neighbourhood{}^{(1)}(i)$ and $\neighbourhood{}^{(2)}(i)$. 
\end{definition}

The left hand panel of Figure~\ref{fig:nom-stab} below provides an example of a \biClique.
We assume that the public goods game with nominations for each $i$ given by $\nominateSet{i}{\scriptscriptstyle{\leq}}$ is played over this network. We further assume that $\kappa(i) = |N_G(i)|$ for each player $i$.
%We further $\kappa(i) = |N_G(i)|$ for every vertex $i$.
Both the middle panel and right hand panel present specialised equilibria. 
In the equilibrium of the middle panel those in $D$ share maximally, while in that of the right hand panel those in $D$ do not (we omit nominations by the non-specialists in both panels to reduce clutter).
It is easy to see that the specialised equilibrium in the middle panel will not change with a unilateral deviation in nomination.
For the equilibrium in the right hand panel however, adding a feasible nomination between, for example, the two specialists on the left hand side, will lead to a new action profile under the best-action response dynamic.

\begin{figure}[h]
\centering
\includegraphics[scale=0.4]{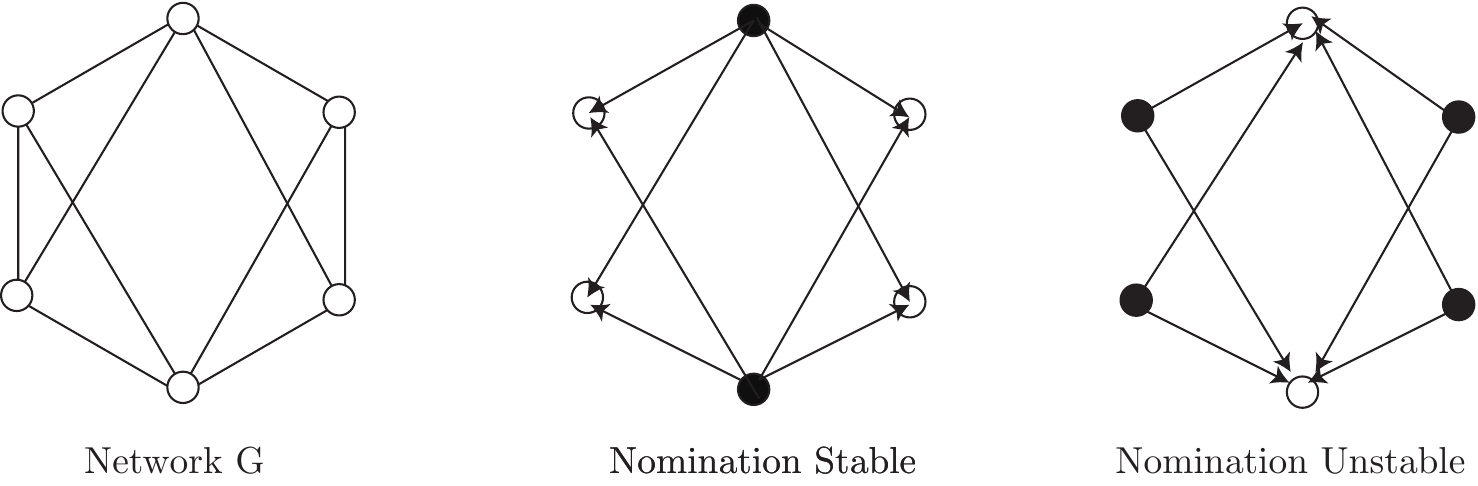}
\caption{A \biClique~and two specialised equilibria.} \label{fig:nom-stab}
\end{figure}

While the graph of Figure~\ref{fig:nom-stab} is symmetric, as are the remaining examples we give of \biClique s, this is not a requirement for a graph to satisfy Definition~\ref{def:biclique}.

We now show, in Proposition~\ref{prop:full-sat} that the feature of nomination stability above is true for all \biClique s and is not a feature of the particular example in Figure~\ref{fig:nom-stab}.
Moreover, the proposition confirms that if $(\sbold, \nominateBold)$ is a specialised Nash equilibrium that is nomination-stable, then all specialists nominate all of their neighbours (hence the full saturation of their capacity constraints).

\begin{proposition} \label{prop:full-sat}
Suppose that $G$ is a \biClique. 
Let $\nominateSet{i}{\scriptscriptstyle{\leq}} := \{ S \subset N_G(i): |S| \leq \kappa(i) \}$ and suppose that $ \kappa(i) =|N_G(i)|$ for all $i$. Suppose that $(\sbold, \nominateBold)$ is a specialised profile induced by $DP$-Nash subgraph $H$ of $G$ and is nomination stable.
Then 
\[
	|m_i| = \shareNumber(i)  \text{ for all } i \in D 
\]
\end{proposition}

\begin{proof}
Let  $(\sbold, \nominateBold)$ be a nomination stable specialised equilibrium induced by $DP$-Nash subgraph $H$ of $G$.
% that is nomination stable.
Consider some player $i \in P$.
We first show that $ |D \cap N^{(1)}(i)| \leq 1$ and $ |D \cap N^{(2)}(i)| \leq 1$. Suppose that  $|D \cap N^{(1)}(i)| \geq 2$. Let $p, q \in D \cap N^{(1)}(i)$ such that $p \neq q$.
Then since $(\sbold, \nominateBold)$ is a Nash equilibrium, $p  \not \in m_q$ and $q \not \in m_p$. Note that $p \in N_G(q)$ and $q \in N_G(p)$ because of $G$ is a \biClique. Thus, if $\sbold' := \mathcal{B}_{\nominateBold + ( p \rightarrow q)}(\sbold)$, action profile $\sbold'$ does not settles in $\sbold$ since the best response of individual $q$ is 0 under $\nominateBold + ( p \rightarrow q)$. Hence, $(\sbold,\nominateBold)$ is not nomination stable which is a contradiction. The same argument holds for $|D \cap N^{(2)}(i)| \geq 2$.

Now suppose that for some $i \in D$, $|m_i| < \kappa(i) $. Then, there exist $j \in N_G(i)$ such that $ j \not\in m_i$. 
There are two cases to consider.

\noindent \textbf{Case 1:}  $j \in D$.
Then let $\sbold' = \mathcal{B}_{\nominateBold + (i \rightarrow j)}(\sbold)$. Then, since the best response of the agent at $j$ is 0 under $\nominateBold + ( i \rightarrow j)$,  $\sbold'$ cannot settle at $\sbold$.

\noindent\textbf{Case 2:} $j \in P$.  Then, $ |D \cap N^{(1)}(j)| \leq 1$ and $ |D \cap N^{(2)}(j)| \leq 1$.
Without loss of generality, suppose that $i \in N^{(1)}(j)$.
Then since $(\sbold,\nominateBold)$ is a Nash equilibrium and $i \not \rightarrow j$, there exists $k \in D \cap N^{(2)}(j)$ such that $j \in m_k$.  Let $\sbold'=\mathcal{B}_{m - (k \rightarrow j)}(\sbold)$. Then, since $|D \cap N^{(2)}(j)| \leq 1$, the best response of the agent at $j$ is $q^*$ and $\sbold'$ does not settle at $\sbold$.
\end{proof}

The underlying reason for Proposition \ref{prop:full-sat} is that the number of specialists in the nomination stable specialised equilibrium of a \biClique~is neither too great or too small.
If a $D$-set is small, then adding a new nomination destabilises a specialised equilibrium.
Whereas if the $D$-set is too large, then removing an existing nomination destabilises a specialised equilibrium.

%We hope that this intuition for the full saturation of capacity constraints carries over a reason class of networks, hence providing the justification for our assumption of the full saturation.

We conclude by noting that a class of graphs that feature prominently in BK (see Figure 2 and Figure 5 in BK), are special cases of \biClique s.
These are defined as follows.
Let $n$ be an even number.
For each $l$ such that $1 \leq l \leq n/2 -1$, define a graph on $n$ vertices, $G_l$, in which the neighbourhood of every vertex $i$, $N_{G_l}(i)$ is the set $\set{ i-l, i-l+1, \cdots, i -1 } \cup \set{ i+1, i+2, \cdots, i + l }$.
Such graphs are \biClique s since for every vertex $i$ we can define
\[
	N^{(1)}(i):= \{ i-l, i-l+1, \cdots, i -1 \} , \quad N^{(2)}(i)=\{ i+1, i+2, \cdots, i + l \}.
\]
When $l=1$, $G_l$ is a ring network (the left panel in Figure \ref{fig:s-net}) and as $l$ increases, the network approaches the complete network.

% (i) vertices consist of $V=\{1, 2, \cdots, N \}$ and (ii) for each $i \in V$, $N_G(i) = \{ i-l, i-l+1, \cdots, i -1 \} \cup \{ i+1, i+2, \cdots, i + l \}$, where we identify two vertices $i, j$ satisfying $i =_{mod \,\, N} j$. Let 
%\[
%	N^{(1)}(i):= \{ i-l, i-l+1, \cdots, i -1 \} , \quad N^{(2)}(i)=\{ i+1, i+2, \cdots, i + l \}.
%\]
%We let $\mathcal{G}_N=\{ G_1, G_2, \cdots, G_{N/2-1} \}$. Then all graphs in $\mathcal{G}_ N$ satisfy Defintion \ref{def:graphs}.
%\noindent Note that if $l=1$, the above graph becomes a ring network (the left panel in Figure \ref{fig:s-net}) and as $l$ increases, the network is getting closes to the complete network. B \& K (2007) use these networks to study the effect of increasing integration within networks (see Figure 5 in B \& K (2007)).
%$\square$

\begin{figure}[h]
\centering
\includegraphics[scale=0.35]{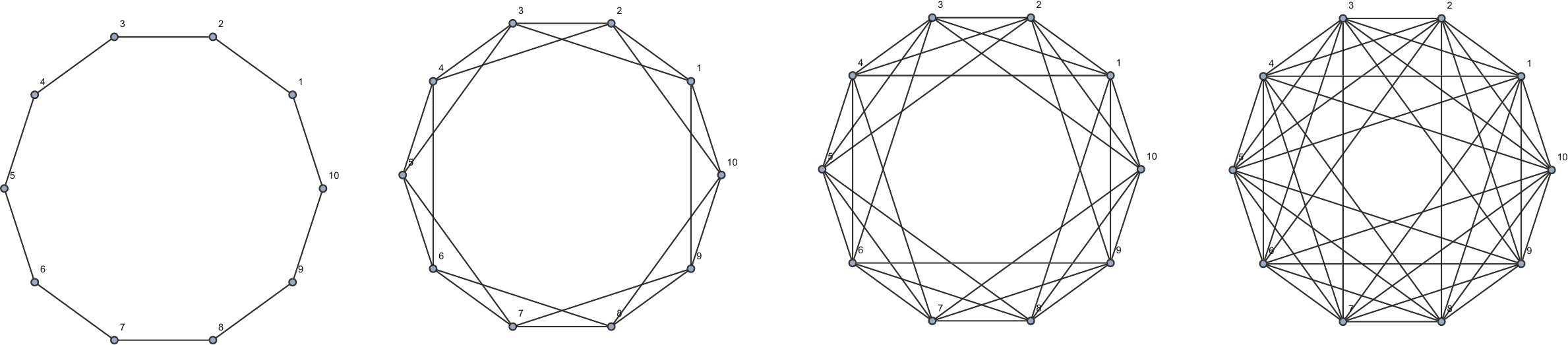}
\caption{With $n=10$, from left to right, we depict $G_{l}$ for $l=1,2,3,4$. } \label{fig:s-net}
\end{figure}

%   Let $M_i := \{ S \subset N_G(i): |S| \leq k \}$ where $2\leq k \leq 2 l$.

%%% ----------------------------------------------------------------------
%%% ----------------------------------------------------------------------
%%% ----------------------------------------------------------------------
%%% ----------------------------------------------------------------------

%\newpage

\section{Conclusion}\label{CONCLUSION}

%%% ----------------------------------------------------------------------
%%% ----------------------------------------------------------------------

This paper introduces a network model of public goods in which individuals face a capacity constraint on how many neighbours they can share with.
Each individual must decide not only how much of the good to provide but also which subset of his neighbours to offer access to.

%Maybe insert some of what referee 2 for econometrica said here.

We focus on a subclass of pure strategy equilibria, that, following  \cite{BramoulleKranton:2007:JET}, are termed specialised equilibria, wherein each individual provides either the desired quantity or provides nothing and free rides.
We prove the existence of specialised equilibria and we show that as capacity increases equilibria need not become more efficient.
Finally we show that specialised equilibria are the only candidates for long run behaviour.

In our model there is a pre-specified network and public good provision occurs over an endogenously generated subnetwork.
But our equilibrium existence result goes through even if there was initially no network (see Appendix \ref{APP:endogenous}).
That is, a special case of our model is one of a public goods game being played over an endogenously generated network.
% model has a second interpretation wherein the nomination component is viewed as pure network formation, with public good provision then occurring over the endogenously generated subnetwork.
When we introduce dynamics, we do so first with the nomination component held fixed so that the endogenously generated subnetwork is constant, and second with only single changes in nomination allowed.
Relaxing this to allow for arbitrary changes in nomination would yield a rich model capturing the coevolution of behaviour and network formation.

In our opinion, individuals being constrained in the number of neighbours they may share with and interact with seems well-suited to applications.
Endogenous network formation with follow-up strategic interactions occurring over the endogenously generated network appears a fruitful research direction.
In the current set up sharing bestows a benefit on neighbours but this need not be the case.
\cite{GutinHirano:2021:JEIC} add constrained sharing to the Susceptible-Infected-Removed (SIR) model of disease transmission of \cite{KermackMcKendrick:1927:PRSLSCPMPC}.
\cite{GutinHirano:2021:JEIC} interpret constrained sharing as ``social distancing'' restrictions imposed on a population and document how the reach of an epidemic is curtailed when such measures are in place.
We leave the study of further variants of interactions on networks with constraints on sharing to future work.

%%% ----------------------------------------------------------------------
%%% ----------------------------------------------------------------------
%%% ----------------------------------------------------------------------

\newpage

\appendix

\setstretch{1.1}

{\footnotesize

\section*{APPENDIX}\label{APP}

\section{Public goods games on endogenous networks}\label{APP:endogenous}

%\section{Two alternative models\protect\footnote{We thank the associate editor and an anonymous referee for encouraging us to write this section.}}\label{APP:Alternative}

Here we consider the special case of our model where the underlying network $G$ is a complete graph.
We will show that this special case is in fact equivalent to a setting where there is initially no network and individuals play a public goods game over an endogenously generated subnetwork.
Our existence result, Theorem~\ref{theorem:balancedPSNE}, ensures existence under this interpretation too.
The original exogenously specified network can be viewed as constrains on the set of individuals that each individual is allowed to nominate. We note that this relationship is symmetric; whenever $i$ can nominate $j$, then $j$ can also nominate $i$.

%In the first subsection, \ref{APP:weakCapacity}, we consider a model wherein each individual $i$ must only nominate at most $\shareNumber(i)$ neighbours and not exactly $\shareNumber(i)$ neighbours like in our model.
%We will show that specialised equilibria in this nearby model that are not specialised equilibria in our model are either (i) not robust against a particular kind of tremble, or (ii) Pareto inefficient when the benefit component of utility, $f(\cdot)$ in \eqref{eq:utility}, is strictly increasing.

%In the second subsection, \ref{APP:endogenous}, we consider the special case of our model where the underlying network $G$ is a complete graph.
%We will show that this special case is in fact equivalent to a setting where there is initially no network and individuals play a public goods game over an endogenously generated subnetwork.
%Our existence result, Theorem~\ref{theorem:balancedPSNE}, ensures existence under this interpretation too.

%We have assumed a model in which there is an exogenously specified network where nominations and action choices result in a public goods game being played over an endogenously generated subnetwork.
%But the original exogenously specified network can be viewed not as an existing network per se, but rather as a constraint on the set of individuals that each individual is allowed to nominate. 

Under this interpretation, the particular case of our model where the exogenous network is a complete graph is akin to a model of endogenous network formation with a public goods game being played over the resulting endogenously generated network.
That is, each individual, call them $\ell$, can nominate \textbf{any} $\shareNumber(\ell)$ individuals as co-benefactors.
So by Theorem~\ref{theorem:balancedPSNE} a pure strategy Nash equilibrium always exists under this interpretation as well.

When the underlying network is not a complete graph, it is easy to see that every pure strategy Nash equilibrium is also an equilibrium when the underlying network is a complete graph.
However, there may be equilibria when the underlying network is a complete graph that are not equilibria otherwise.
To see this, consider a 5-individual instance of the endogenous network formation variant (i.e., $G$ is a complete graph), with players labelled $h, i, j, k$, and $\ell$.
(This instance will parallel that of Example~\ref{ex:Netflix}.)
We assume that capacities are given by $\shareNumber(h) = \shareNumber(i) =\shareNumber(j) = \shareNumber(k) = 1$ and $\shareNumber(\ell) = 4$.

%We conclude by noting that there may be many more equilibria when the initial network is complete. with no initial network is a richer environment than our model.
%In particular, while every equilibrium of our model is an equilibrium of the endogenous network formation variant, there may be equilibria of the new variant that are not equilibria in our model.
%To see this, consider a 5-individual instance of the endogenous network formation model, with players labelled $h, i, j, k$, and $\ell$.
%(This instance will parallel that of Example~\ref{ex:Netflix}.)
%We assume that capacities are given by $\shareNumber(h) = \shareNumber(i) =\shareNumber(j) = \shareNumber(k) = 1$ and $\shareNumber(\ell) = 4$.

%We suppose initially that there is no network.
%As in Example~\ref{ex:Netflix} individuals choose who to nominate and choose action 0 or 1.

Clearly both specialised equilibria in Figure~\ref{fig:starEquilibria} are also specialised equilibria in this endogenous network formation variant.
However, when the network is complete so that each individual can nominate \textbf{any} subset of individuals equal to their capacity, there are additional equilibria.
One such equilibrium that it not an equilibrium for the exogenous network of Example~\ref{ex:Netflix} is depicted below in Figure~\ref{fig:starEquilibriaAPP}

\begin{figure}[hbt!]
\centering
\tikzstyle{vertexX}=[circle,draw, top color=gray!10, bottom color=gray!70, minimum size=10pt, scale=0.9, inner sep=0.5pt]
\tikzstyle{vertexY}=[circle,draw, top color=black!10, bottom color=red!70, minimum size=25pt, scale=0.8, inner sep=0.4pt]
\tikzstyle{vertexZ}=[circle,draw, top color=black!10, bottom color=blue!70, minimum size=25pt, scale=0.8, inner sep=0.4pt]
\tikzset{arc/.style = {->,> = latex'}}
%\hfill
%\begin{tikzpicture}[scale=0.5]
%%\draw (5,-1.5) node {{\small $(G_1,\kkk_1)$}};
%%\draw (1.0,1.0) node {{\small $1$}};
%%\draw (1.0,6.0) node {{\small $1$}};
%%\draw (9.0,6.0) node {{\small $1$}};
%%\draw (9.0,1.0) node {{\small $1$}};
%%\draw (5.0,4.9) node {{\small $3$}};
%%\node (x) at (3.0,1.0) [vertexY] {$x$}; 
%%\node (z) at (7.0,1.0) [vertexY] {$z$}; 
%\node (y) at (5.0,5.0) [vertexZ] {$\ell$};
%\node (x1) at (2.5,2.5) [vertexY] {$k$}; 
%\node (x2) at (2.5,7.5) [vertexY] {$h$};
%\node (x3) at (7.5,7.5) [vertexY] {$i$}; 
%\node (x4) at (7.5,2.5) [vertexY] {$j$}; 
%%\draw [thick] (y) -- (x1); 
%%\draw [thick] (y) -- (x2); 
%%\draw [thick] (y) -- (x3); 
%%\draw [thick] (y) -- (x4); 
%%\node (x1) at (2.5,1.0) [vertexY] {$x_{1}$}; 
%%%\node (x2) at (2.5,6.0) [vertexY] {$x_{1}$};
%%\node (x3) at (5.0,6.6) [vertexY] {$x_{2}$}; 
%%\node (x4) at (7.5,1.0) [vertexY] {$x_{3}$}; 
%%\node (y) at (5.0,3.5) [vertexZ] {$x_{4}$};
%\draw [arc,line width=1.3pt] (y) -> (x1); 
%\draw [arc,line width=1.3pt] (y) -> (x2); 
%\draw [arc,line width=1.3pt] (y) -> (x3); 
%\draw [arc,line width=1.3pt] (y) -> (x4); 
%%\draw [thick] (y) -- (x1); 
%%\draw [thick] (y) -- (x2); 
%%\draw [thick] (y) -- (x3); 
%%\draw [thick] (y) -- (x4); 
%%\draw [thick] (z) -- (y2); 
%%\draw [thick] (z) -- (y3); 
%\end{tikzpicture} \hfill 
\begin{tikzpicture}[scale=0.5]
\node (y) at (5.0,5.0) [vertexY] {$\ell$};
\node (x1) at (2.5,2.5) [vertexY] {$k$}; 
\node (x2) at (2.5,7.5) [vertexZ] {$h$};
\node (x3) at (7.5,7.5) [vertexZ] {$i$}; 
\node (x4) at (7.5,2.5) [vertexZ] {$j$}; 
%\draw [thick] (y) -- (x1); 
%\draw [thick] (y) -- (x2); 
%\draw [thick] (y) -- (x3); 
%\draw [thick] (y) -- (x4); 
%\node (x1) at (2.5,1.0) [vertexY] {$x_{1}$}; 
%%\node (x2) at (2.5,6.0) [vertexY] {$x_{1}$};
%\node (x3) at (5.0,6.6) [vertexY] {$x_{2}$}; 
%\node (x4) at (7.5,1.0) [vertexY] {$x_{3}$}; 
%\node (y) at (5.0,3.5) [vertexZ] {$x_{4}$};
%\draw [arc,line width=1.3pt] (x1) -> (y); 
\draw [arc,line width=1.3pt] (x2) -> (x1); 
\draw [arc,line width=1.3pt] (x3) -> (y); 
\draw [arc,line width=1.3pt] (x4) -> (y); 
%\draw (5,-1.5) node {{\small $(G_1,\kkk_1)$}};
%\draw (1.0,1.0) node {{\small $1$}};
%\draw (1.0,6.0) node {{\small $1$}};
%\draw (9.0,6.0) node {{\small $1$}};
%\draw (9.0,1.0) node {{\small $1$}};
%\draw (5.0,4.9) node {{\small $3$}};
%\node (x) at (3.0,1.0) [vertexY] {$x$}; 
%\node (z) at (7.0,1.0) [vertexY] {$z$}; 
%\node (x1) at (2.5,1.0) [vertexZ] {$x_{1}$}; 
%%\node (x2) at (2.5,6.0) [vertexY] {$x_{1}$};
%\node (x3) at (5.0,6.6) [vertexZ] {$x_{2}$}; 
%\node (x4) at (7.5,1.0) [vertexZ] {$x_{3}$}; 
%\node (y) at (5.0,3.5) [vertexY] {$x_{4}$};
%%\draw [arc,line width=1.3pt] (y) -> (x1); 
%
%\draw [thick] (y) -- (x1); 
%%\draw [thick] (y) -- (x2); 
%\draw [thick] (y) -- (x3); 
%\draw [thick] (y) -- (x4); 
%%\draw [thick] (z) -- (y2); 
%%\draw [thick] (z) -- (y3); 
\end{tikzpicture} 
\caption{Equilibrium of endogeneous network formation model.}
\label{fig:starEquilibriaAPP}
\end{figure}

\noindent
where the difference with the equilibrium depicted in the right hand panel of Figure~\ref{fig:starEquilibria} is that now $h$ is nominating $k$ instead of $\ell$ and so clearly it is no longer optimal for $k$ to buy.
This is permitted when the underlying graph is complete since $h$ and $k$ are neighbours, whereas $h$ and $k$ are not neighbours in the social network of Example~\ref{ex:Netflix} (see Figure~\ref{fig:5star}).

\section{Results and proofs not in main text}\label{App:Lemmas}

The proofs of Propositions~\ref{proposition:necessary} and \ref{proposition:notSufficient} and the proof of Theorem~\ref{theorem:sufficient} require some lemmas.
We open this section by proving the necessary lemmas, but before starting with that we introduce some notation. 
We are interested in comparing population level contributions. As such we will wish to order action profiles wherever possible. For any two action profiles $\sbold, \sbold' \in \Sbold$, we say $\sbold \geq \sbold'$ if $\strategy{i} \geq \strategy{i}'$ for all $i \in \set{1, \dots, n}$, and $\sbold > \sbold'$ if $x_{i} \geq y_{i}$ for all $i \in \set{1, \dots, n}$ and $\strategy{j} > y_{j}$ for at least one $j \in \set{1, \dots, n}$.

%, and $\sbold \gg \sbold'$ if $\strategy{i} > \strategy{i}$ for all $i \in \set{1, \dots, n}$.

%\stefanie{\marginpar{some deletion here}}
%\stefanie{
%With this, we have the following lemma whose proof is immediate from the definition bet-action reply from \eqref{eqn:bestResponse}.
%
%\begin{lemma}\label{lemma:ordering}
%Consider any two action profiles $\sbold$ and $\sbold'$ such that $\sbold \leq \sbold'$. Then, given nominating profile $\nominateBold$, we have
%\begin{enumerate}
%\item
%$\bestReply{\nominateBold}{\sbold} \geq \bestReply{\nominateBold}{\sbold'}$, and
%\item
%$\bestReplyRepeat{\nominateBold}{2}(\sbold) \leq \bestReplyRepeat{\nominateBold}{2}(\sbold')$
%\end{enumerate}
%\end{lemma}
%
%\red{Philip does not like the following short paragraph nor Definition \ref{definition:sequences}. Come up with way to say both parts more clearly and concisely.}
%
%Ultimately, with a process defined by repeated applications of the best-action reply, we wish to the stability of various action profiles. To this end, we introduce sequences defined as follows.
%
%\begin{definition}\label{definition:sequences}
%Given action profile $\sboldstar$, define a sequence of action profiles around $\sboldstar$ as follows. Choose $\sbold^{(0)}$ such that $\sbold^{(0)} < \sboldstar$, and recursively define $\sbold^{(t)} = \bestReply{m}{\sbold^{t-1}}$ for all $t \geq 1$.
%\end{definition}
%}
We then have the following.

\begin{lemma}\label{lemma:nashSequencePopulation}
Suppose that $(\sboldstar, \nominateBoldStar)$ is a pure strategy Nash equilibrium and let $\sbold\leq \sboldstar$. Then the best action evolution of $\sbold$ satisfies for all $t\geq 0$
\begin{description}
\item[(i)]
$\sbold^{(t+1)} \geq \sboldstar \geq \sbold^{(t)}$ if $t$ is even.
\item[(ii)]
$\sbold^{(t+1)} \leq \sboldstar \leq \sbold^{(t)}$ if $t$ is odd.
\end{description}
\end{lemma}

\begin{proof}
%By definition, $\sbold^{(0)} \leq \sboldstar$. From Lemma \ref{lemma:ordering}, it is enough to show that $\sbold^{(1)} \geq \sboldstar$. By the definition of the sequence and the best-action reply (equation \eqref{eqn:bestResponse}), for every individual $i$ we have that
%\[
%\strategy{i}^{(1)} = \qstar - \sum_{ \{ j \in \neighbourhood{G}(i): i \in \nominate{j} \}} \strategy{j}^{(0)}
%\]
Let $t=0$. Then by assumption $\sbold^{(0)}\leq \sboldstar$.
If we choose $i$ to be an individual with $\strategy{i}^{*} =0$, then clearly we have $\strategy{i}^{(1)} \geq 0 = \strategy{i}^{*} =0$. Now, let $i$ be an individual with $\strategy{i}^{*} \neq 0$. Then
\[ \strategy{i}^{*} = \qstar - \sum_{ \{ j \in \neighbourhood{G}(i) : i \in \nominate{j} \}} \strategy{j}^{*}
 \leq \qstar - \sum_{ \{ j \in \neighbourhood{G}(i) : i \in \nominate{j} \}} \strategy{j}^{(0)}
= \strategy{i}^{(1)}.
\]
Thus $x_i^{(1)} \geq x_i^*$ and  $\sbold^{(1)} \geq \sboldstar$. By considering  individuals with $x^{(2)}\not= 0$ separately, we can use the same argument  to shows that $\sbold^{(2)} \leq \sboldstar$. Thus we have $\sbold^{(2)} \leq  \sboldstar \leq \sbold^{(1)}$. The result then follows easily by induction.
\end{proof}

Lemma \ref{lemma:nashSequencePopulation} implies that if we take a pure strategy Nash equilibrium action profile $\sboldstar$, for any $\sbold \leq \sboldstar$ the best action evolution of $\sbold$ either oscillate around $\sboldstar$ forever, or there exists an $t_0$ such that for all $t\geq t_0$,  $\sbold^{(t)}=\sboldstar$.

% Note that this is the case as we consider only finite networks and so there are only a finite number of different possibilities for the $\sbold^{t}$ and therefore they have to repeat.  This motivates the following definition.

%\begin{definition}
%We say that the best action evolution of $\sbold$ settles in $\sboldstar$ if there exists an $t_0$ such that for all $t\geq t_0$,  we have $\sbold^{(t)}=\sboldstar$. 
%\end{definition}

We emphasise that Lemma \ref{lemma:nashSequencePopulation} holds for all Nash equilibrium strategy profiles and not simply specialised ones. While Lemma \ref{lemma:nashSequencePopulation} is a statement about the population action profiles, the next two results, Lemmas \ref{lemma:nashSequenceIndividual*} and \ref{lemma:nashSequenceIndividual} are statements about best-action responses at the level of the individual. We emphasise again that both results apply to all pure strategy Nash equilibria, not simply specialised ones.
The first lemma is an immediate corollary of Lemma~ \ref{lemma:nashSequencePopulation}.
\begin{lemma}\label{lemma:nashSequenceIndividual*}
Suppose that $(\sboldstar, \nominateBoldStar)$ is a pure strategy Nash equilibrium. Let $\sbold < \sboldstar$. Then the best action evolution of $\sbold$ satisfies the following:
\begin{description}
\item[(i)]
If $t$ is odd, then $\strategy{i}^{(t)} \neq 0$, for all $i$ such that $\strategy{i}^{*} > 0$, and 
\item[(ii)]
If $t$ is even, then $\strategy{i}^{(t)} \neq \qstar$, for all $i$ such that $\strategy{i} < \qstar$.
\end{description}
\end{lemma}
%\begin{proof}
%Suppose for a contradiction, that $t$ is odd and $\strategy{i}^{(t)} = 0$ for some $i \in V$ such that $\strategy{i}^{*} > 0$. But note that Lemma \ref{lemma:nashSequencePopulation} then implies that $\strategy{i}^{*} = 0$ which is a contradiction to the fact that $\strategy{i}^{*} > 0$.%Similarly but opposite if $t$ is even. That is, if $t$ is even and $\strategy{i}^{(t)} = \qstar$ for some $i \in V$, then $\strategy{i}^{*} = \qstar$, which is again a contradiction.
%\end{proof}

We will need the following lemma which essentially says that a player, say player $\ell$, reacts to any changes in action of players that influence it, unless $\ell$ does not need to provide anything. That is, $x_{\ell}$ can remain zero and is still satisfied.

%\red{We will need the following lemma which  essentially says that a state $\ell$ reacts to any changes of states that affect it unless $\ell$ does not need to provide anything, that is $x_\ell$ can remain zero and is still satisfied.}

\begin{lemma}\label{lemma:nashSequenceIndividual}
Suppose that $(\sboldstar, \nominateBoldStar)$ is a pure strategy Nash equilibrium. 
Let $\sbold\leq \sboldstar$ be an action profile and consider the best action evolution of $\sbold$. Let  $\ell$ be nominated by $i$, that is $\ell \in \nominate{i}$.
%such that $\strategy{\ell}^{(t)} > 0$.%and $\strategy{\ell}^{(t+1)} > 0$. 
\begin{align}
\text{If } \strategy{\ell}^{(t)} > 0 \mbox{ and }\strategy{i}^{(t)} > \strategy{i}^{(t-1)}, &\text{ then } \strategy{\ell}^{(t+1)} < \strategy{\ell}^{(t)} \label{eqn:nashSequenceIndividual1}\\
\text{If } \strategy{\ell}^{(t+1)} > 0  \mbox{ and }\strategy{i}^{(t)} < \strategy{i}^{(t-1)}, &\text{ then } \strategy{\ell}^{(t+1)} > \strategy{\ell}^{(t)} \label{eqn:nashSequenceIndividual2}
\end{align}
\end{lemma}
\begin{proof}
We first show \eqref{eqn:nashSequenceIndividual1}. So assume  $\strategy{i}^{(t)} > \strategy{i}^{(t-1)}$.
%and $\strategy{i}^{*} > 0$.
Then, by Lemma \ref{lemma:nashSequencePopulation}, $t$ must be odd. Now, again using Lemma~ \ref{lemma:nashSequencePopulation}, we have for  $\ell \in \nominate{i}$
\begin{align*}
\sum_{\set{j \in \neighbourhood{G}(\ell) : \ell  \in \nominate{j}}} \strategy{j}^{(t)} &- \sum_{\set{j \in \neighbourhood{G}(\ell) : \ell  \in \nominate{j}}} \strategy{j}^{(t-1)}\\
&= \strategy{i}^{(t)} - \strategy{i}^{(t-1)} + \sum_{\set{j \in \neighbourhood{G}(\ell) \, : \, \ell  \in \nominate{j}, \, j \neq i}}\strategy{j}^{(t)} - \sum_{\set{j \in \neighbourhood{G}(\ell) \, : \, \ell  \in \nominate{j}, \, j \neq i}} \strategy{j}^{(t-1)}\\
&\geq \strategy{i}^{(t)} - \strategy{i}^{(t-1)}\\
&> 0
\end{align*}
Rearranging and adding $\qstar$ to both sides of this inequality gives that 
\[
\qstar - \sum_{\set{j \in \neighbourhood{G}(\ell) : \ell  \in \nominate{j}}} \strategy{j}^{(t)} < \qstar - \sum_{\set{j \in \neighbourhood{G}(\ell) : \ell  \in \nominate{j}}} \strategy{j}^{(t-1)}.
\]
But note that by the best-action response  \eqref{eqn:bestResponse} and because we assume that $\strategy{\ell}^{(t)}>0$ the right hand side is equal to $\strategy{\ell}^{(t)}$. Also $\strategy{\ell}^{(t+1)}$ equals either the left hand side of the above inequality or is equal to zero and in both cases is  smaller than $\strategy{\ell}^{(t)}$.  The inequality in \eqref{eqn:nashSequenceIndividual2} is shown in a similar manner.
\end{proof}

We are now in a position to prove Propositions~\ref{proposition:necessary} and \ref{proposition:notSufficient}. For convenience we restate both.

\begin{customthm}{2}
Suppose that $(\sboldstar, \nominateBoldStar)$ is a pure strategy Nash equilibrium such that $0 < \strategy{\ell}^{*} < \qstar$ for some $\ell$. Then $(\sboldstar, \nominateBoldStar)$ is not action-stable.
\end{customthm}

\begin{proof}
Define $L := \set{j : 0 < \strategy{j}^{*} < \qstar}$. By assumption $\ell \in L$.  First we will show that $\ell$ cannot be the only element in the set $L$.  This is immediate since 
\[
0 < \strategy{\ell}^{*} = \qstar - \sum_{\set{j \in \neighbourhood{G}(\ell) \, : \, \ell \in \nominate{j}^{*}}}\strategy{j}^{*} < \qstar  \,\,\,\, \Longrightarrow \,\,\,\, 0 < \sum_{\set{j \in \neighbourhood{G}(\ell) \, : \, \ell \in \nominate{j}^{*}}}\strategy{j}^{*} < \qstar.
\]
Thus, there exists some $j \in L$ such that $j \in \neighbourhood{G}(\ell)$, $\ell \in \nominate{j}^*$.  It follows that there exist indices $i_1, i_2, \ldots, i_k$ such that $i_{1} \in \nominate{i_{2}}^*, i_{2} \in \nominate{i_{3}}^*, \ldots, i_{k-1}\in \nominate{i_k}^*$ and $i_k \in \nominate{i_{1}}^*$. (For the readers familiar with graph theory we can build a directed graph on $L$ with an arc from $j$ to $i$ if $i\in \nominate{j}^*$. The above result implies that every vertex has at least one in-neighbour and hence the graph contains a cycle.)
 By renaming if necessary we may assume that $i_j=j$ for all $j=1,\ldots,k$, so 
for $i=1,\ldots,k$ we have $i\in L$, $i\in \nominate{i+1}^*$ and $k\in \nominate{1}^*$.
We define $\sbold$ by subtracting $1$ from the first coordinate of $\sbold^*$, that is, 
\[ \sbold= (x^*_1-1,x^*_2\ldots,x^*_n).\] 
%Note that $\strategy{1}^{(0)} = x_1^*-1\not= x_1^*=\strategy{1}^{1}$.
We claim that $\sbold$ does not settle in $\sboldstar$. We prove the claim by contradiction. In particular we assume that there exists a (minimal) $t_0$ such that for all $t\geq t_0$  and for all $i=1,\ldots,k$, we have $x_i^{(t)}=x_i^{(t+1)}$. Note that $t_0\geq 1$ as  $\strategy{1}^{(0)} = x_1^*-1\not= x_1^*=\strategy{1}^{(1)}$. Also 
for all $i=1,\ldots,k$, we have $\strategy{i}^{(t_0)}=\strategy{i}^{(t_0+1)}=\strategy{i}^{*}>0$. The contradiction now follows from Lemma~\ref{lemma:nashSequenceIndividual} by choosing $i\in \{1,\ldots,k\}$ such that 
$x_i^{t_0-1}\not=x_i^{t_0}$ and $\ell=i+1$. Such an $i$  exists as we have chosen $t_0$ minimal. 
\end{proof}

\begin{customthm}{3}
Suppose $(\sboldstar, \nominateBoldStar)$ is a specialised Nash equilibrium induced by $DP$-Nash subgraph $H$ of $G$, where (i) $\nominate{i} \cap \neighbourhood{H}(i) \neq \emptyset$ for all $i \in P$, and (ii) $\degree{H}(i) = 1$ for some $i \in P$.
Then $(\sboldstar, \nominateBoldStar)$ is not action-stable.
\end{customthm}

\begin{proof}
Without loss of generality we may assume that $1\in P$, $d_{H}(1) = 1$ and that $2$ is the neighbour of $1$ in $H$ (i.e., $1 \in m_2$). Also, since $m_1 \cap N_H(1) \neq \emptyset$ and $|N_H(1)|=d_H(1)=1$, we have $N_H(1) \subset m_1$, and thus $2 \in m_1$. 
Note that $x_2^*=\qstar$ as the strategy profile is a specialised equilibrium.
Consider 
\[\sbold=(x^*_1,x^*_2-1,\ldots,x^*_n)=(x^*_1,\qstar-1,\ldots,x^*_n).\] 

We claim that $\sbold$ does not settle in $\sboldstar$. To see this we prove by induction that for all odd $t$ we have $x_1^{(t)} > 0$    and for all even $t$ we have $x_2^{t}<q^*$, 
Clearly this is true for $t=0$ and also for $t=1$ as the best action response for $1$ is $x^{(1)}_1=1>0$. Now assume that the statement is true for all $t\leq t_0$ with $t_0\geq 2$.

 If $t_0$ is even then by induction hypothesis $x_2^{t_0}<q^*$ and by Lemma~\ref{lemma:nashSequencePopulation}  for all $j\in P$ we have $x_j^{(t_0)}\leq x^*_j = 0$ and thus $x_j^{(t_0)} = 0$. Since $1$ is only nominated by $2$ and perhaps elements in $P$,  the best action response for $1$ is $x_1^{(t_0+1)}=\qstar-x_2^{(t_0)}>0$.
 
  If $t_0$ is odd then by induction hypothesis $x_1^{t_0}>0$ and by Lemma~\ref{lemma:nashSequencePopulation}  for all $j\in D$ we have $x_j^{(t_0)}\geq x^*_j = \qstar$ and thus $x_2^{(t_0)} = \qstar$. Since $2$ is nominated by $1$ the best action response for $2$ is $x_2^{(t_0+1)}\leq \qstar-x_1^{(t_0)}<\qstar$.
\end{proof}

%%% ----------------------------------------------------------------------
%%% ----------------------------------------------------------------------
%%% ----------------------------------------------------------------------
%%% ----------------------------------------------------------------------

\section{Properties of $D$-sets}\label{App:DSets}

%%% ----------------------------------------------------------------------
%%% ----------------------------------------------------------------------

First we show some upper and lower bounds on the sizes of $D$-sets for all graphs. Subsequently we look at specific classes of graphs that are commonly studied. The following definitions are required. For a $G = (V, E)$ be a graph and capacity function $\shareNumber : V \to \Natural_{0}$, we define $\underline{\shareNumber} = \min_i \shareNumber(i)$, $\bar \shareNumber = \max_i \shareNumber(i)$, and $\underline{\degree{G}} := \min_i \degree{G}(i)$. Furthermore, letting $\Real$ denote the real line, we define the ceiling and floor functions as follows:\ for any $x \in \Real$, let $\ceiling{x} := \min\set{n \in \Natural \, | \, x \leq n}$ and let $\flooring{x} := \max\set{n \in \Natural \, | \, x \geq n}$.

\subsection{Bounds on $D$-sets}

\begin{lemma} \label{lem:1st}
Let $G = (V, E)$ be a graph and $\shareNumber : V \to \Natural_{0}$ be a capacity function. Then for any $DP$-Nash subgraph $H$ of $G$ we have,
\begin{description}
\item[(i)]
$|D \cup P| \leq (1+ \bar \shareNumber) |D|$
\item[(ii)]\label{goat}
$|D \cup P| \geq |D| +\min \{\underline \shareNumber, \underline{\degree{G}} \}$
\end{description}
\end{lemma}
\begin{proof}
  (i) We have
  \begin{align*}
     |V| &= |D \cup P|\\
     &= |D| +|P| \\
     &\leq |D|+\sum_{ i \in D} \min \{\shareNumber(i), d_G(i) \} \\
     &\leq |D| +\sum_{ i \in D} \shareNumber(i) \\
     &\leq |D|+\bar \shareNumber |D| \\
     &= (1+ \bar \shareNumber ) |D|
  \end{align*}
  (ii) For any $i \in D$ we have that $\neighbourhood{H}(i) \subseteq P$. As such
  \begin{align*}
    |V| &= |D\cup P| \\
    &\geq |D|+\neighbourhood{H}(i) \\
    &= |D|+\min \{\shareNumber(i), \neighbourhood{G}(i) \} \\
    &\geq |D| +\min \{\underline \shareNumber, \underline{\degree{G}} \}
  \end{align*}
\end{proof}

Recall that for a given graph $G$ and given capacity function $\shareNumber$, $\delta^\shareNumber_{\min}(G)$ and $\delta^\shareNumber_{\max}(G)$ represent the size of minimum and maximum $D$-sets taken over all $DP$-Nash subgraphs of $G$. We have the following lemma.
\begin{lemma} \label{lem:2nd}
    Let $G = (V, E)$ be a graph and $\shareNumber : V \to \Natural_{0}$ be a capacity function. Then we have \\
    (i) $ \delta^\shareNumber_{\min}(G) \geq \frac{|V|}{1+ \bar \shareNumber}$. \\
    (ii) $ \delta^\shareNumber_{\max}(G) \leq |V| - \min \{\underline \shareNumber, \underline{\degree{G}} \}$.
\end{lemma}
\begin{proof}
\begin{description}
\item[(i)]
Suppose to a contradiction that $\delta^\kappa_{\min}(G) < \frac{|G|}{1+ \bar \shareNumber}$. Then there exists a $DP$-Nash subgraph of $G$ such that $|V| <  \frac{|G|}{1+ \bar \shareNumber} =  \frac{|D \cup P|}{1+ \bar \shareNumber}$ which is contradiction to part (i) of Lemma \ref{lem:1st}.
\item[(ii)]
Suppose, again to a contradiction, that $\delta^k_{\max}(G) > |V| - \min \{\underline \shareNumber, \underline{\degree{G}} \}$. Then there exists a $DP$-Nash subgraph $H$ of $G$ with $D$-set $D$ such that $|D| > |V| - \min \{\underline \shareNumber, \underline{\degree{G}} \} =  |D\cup P| - \min \{\underline \shareNumber, \underline{\degree{G}}\}$ which is a contradiction to part (ii) of Lemma \ref{lem:1st}.
\end{description}
\end{proof}

{\begin{remark}
  \normalfont For Lemma \ref{lem:2nd}, we need
  \[
    |G| - \min \{ \underline \shareNumber, \underline{\degree{G}} \} \geq \frac{|G|}{1+ \bar \shareNumber} \iff (1+\bar \shareNumber) ( |G| - \min \{ \underline \shareNumber, \underline{\degree{G}} \} ) \geq |G|
  \]
  To show this, first observe that
  \begin{equation}\label{eq:int-con}
    \bar \shareNumber + \bar \shareNumber (|V| - 1) \geq \min \{\underline \shareNumber, \underline{\degree{V}} \} + \bar \shareNumber \min \{ \underline \shareNumber, \underline{\degree{G}} \}
  \end{equation}
  because $\bar \shareNumber \geq \underline \shareNumber$ and $|V|-1 \geq \degree{G}(i)$ for all $i$. By rearranging \eqref{eq:int-con}, we find that
  This follows because
  \[
  \bar \shareNumber |V| - (1+\bar \shareNumber)\min \{ \underline \shareNumber, \underline{\degree{G}}  \} \geq 0 \iff (1+\bar \shareNumber) ( |V| - \min \{ \underline \shareNumber, \underline{\degree{G}} \} ) \geq |V|
  \]
  $\blacksquare$
\end{remark}

Now, putting Lemmas \ref{lem:1st} and \ref{lem:2nd} together, we obtain the following proposition.
\begin{proposition}\label{prop:bound}
Suppose that $\kappa(i) \leq \underline{\degree{G}}$ for some $i$. Then we have
\[
    \Ceiling{\frac{|V|}{1+\bar\shareNumber }} \leq   \delta^\shareNumber_{\min}(G) \leq  \delta^\shareNumber_{\max}(G)  \leq |V| - \underline \shareNumber
\]
\end{proposition}
\begin{proof}
The middle inequality is true by definition, while the left and right inequalities come from Lemma \ref{lem:2nd} parts (i) and (ii) respectively.
\end{proof}

Proposition \ref{prop:bound} holds for all graphs. However, we can improve upon those bounds for various classes of graph. This is the purpose of the following subsection.

\subsection{$D$-Sets for complete graphs}
\noindent
In this subsection we suppose that $G = (V, E)$ is a complete graph and that all individuals have the same capacity so that $1 \leq \shareNumber(i)= k \leq |V| -1$ for all $i \in V$.

Given this we have $\underline{\degree{G}} = |V|-1$ and thus $\min_{i} \{ \shareNumber(i), \underline{\degree{G}} \} = \min\{ \underline \shareNumber, \underline{\degree{G}} \} = k$ for all $i \in V$. Thus from Proposition \ref{prop:bound} we obtain the following bounds,
\begin{equation}\label{eq:comp-bound}
  \Ceiling{ \frac{|V|}{1+k} } \leq   \delta^\shareNumber_{\min}(G)
\leq  \delta^\shareNumber_{\max}(G) \leq |V| - k
\end{equation}
In particular, if  $k= |V|-1$, then
\[
    \delta^\shareNumber_{\min}(G) = \delta^\shareNumber_{\max}(G)=1
\]

We can show that there exist $k$-$DP$-subgraphs that achieve the lower and upper bounds, $\Ceiling{ \frac{|V|}{1+k} }$ and $ |V| - k $. \\

To achieve the lower bound, choose $\Ceiling { \frac{|V|}{1+k} }$ individuals from $G$ and designate these individuals as $D$. Now define let $R= V \setminus D$. Then, using that
\[
        |V| - k \geq \Ceiling{ \frac{|V|}{1+k} } \iff |V| - \Ceiling{ \frac{|V|}{1+k} } \geq k,
\]
and $\Ceiling{ \frac{|V|}{1+k} } \geq 1$ (since $k \leq |V|-1$) and
\[
     \frac{|V|}{1+k} \leq \Ceiling{ \frac{|V|}{1+k} } \iff |V| - \Ceiling{ \frac{|G|}{1+k} } \leq \Ceiling{ \frac{|G|}{1+k} } k
\]
we obtain
\[
    k \leq |R| \leq \Ceiling{ \frac{|V|}{1+k} } k.
\]
That is, there are least $k$ agents in $R$ and at most $\lceil \frac{|V|}{1+k} \rceil k$. Thus, since $G$ is a complete network, we make each agent in $D$ nominate $k$ agents in $R$ such that $D \cup R=G$ and obtain $DP$-Nash subgraph of $G$.

For the upper bound, we choose $|V|-k$ individuals as elements in $D$. Then there exists exactly $k$ remaining individuals that we refer to as $P$. Since $G$ is a complete graph, we make it such that all individuals in $D$ nominate each of the $k$ remaining agents in $P$ and obtain $DP$-Nash subgraph of $G$. Thus we obtain the following proposition.

\begin{proposition}
Suppose that $G = (V, E)$ is a complete graph and $1 \leq \shareNumber(i)=k \leq |G| -1$ for all $i$. Then we have \begin{equation}\label{eq:comp-bound}
  \Ceiling{ \frac{|V|}{1+k} } \leq \delta^\shareNumber_{\min}(G) \leq  \delta^\shareNumber_{max}(G)  \leq |G| - k.
\end{equation}
Moreover there exist bipartite $DP$-Nash subgraphs $\underbar{H}$ and $\bar{H}$ with partite sets  $(\underline D, \underline P)$ and  $(\bar  D, \bar P)$ such that
\[
    |\underline D| = \Ceiling{ \frac{|V|}{1+k} } \text{ and }  |\bar D|=  |V| - k
\]
\end{proposition}

\begin{figure}[ht!]
\centering
\includegraphics{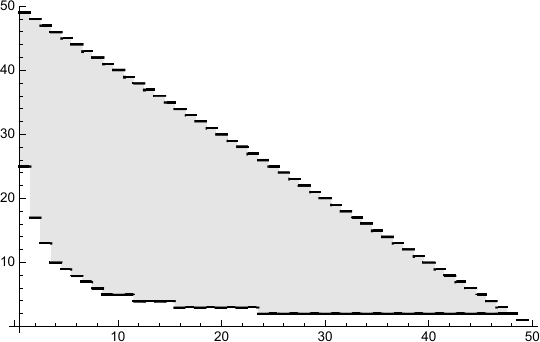}
\caption{The lower and upper bounds for the complete graph with $|V|=50$. Horizontal axis: $k=1, \cdots, 49$. Vertical axis:\ lower and upper bounds}\label{fig:comp-bd}
\end{figure}

}

%%% ----------------------------------------------------------------------
%%% ----------------------------------------------------------------------
%%% ----------------------------------------------------------------------

\newpage
\bibliographystyle{plainnat}
\bibliography{netflix.bib}

%%% ----------------------------------------------------------------------
%%% ----------------------------------------------------------------------
%%% ----------------------------------------------------------------------

\end{document}